\pdfoutput=1
\documentclass[%
 reprint,
 amsmath,amssymb,
 aps,
 nofootinbib,
 superscriptaddress,
 longbibliography,
 preprintnumbers
]{revtex4-2}

\usepackage{amsthm}
\usepackage{graphicx}
\usepackage{xcolor}
\usepackage{braket}
\usepackage{amsmath}
\usepackage{amssymb}
\usepackage{enumerate}
\usepackage{mathtools}
\usepackage{etoolbox}
\usepackage{algorithmic}
\usepackage{algorithm}
\usepackage{bbm}
\usepackage{here}
\usepackage{subfigure}
\usepackage[
colorlinks=true,
urlcolor=blue,
citecolor=blue,
linkcolor=blue,
hyperfootnotes=false,
breaklinks=true
]{hyperref}

\newtheorem{theorem}{Theorem}
\newtheorem{lemma}{Lemma}

\newtheorem{definition}{Definition}
\newtheorem{assumption}{Assumption}

\allowdisplaybreaks

\begin{document}
\preprint{RESCEU-22/23}
\title{Quantum algorithm for the Vlasov simulation of the large-scale structure formation with massive neutrinos}
\author{Koichi Miyamoto}
\email{miyamoto.kouichi.qiqb@osaka-u.ac.jp}
\affiliation{Center for Quantum Information and Quantum Biology, Osaka University, Toyonaka, Osaka 560-0043, Japan}
\author{Soichiro Yamazaki}
\affiliation{Department of Physics, Graduate School of Science, The University of Tokyo, Bunkyo-ku, Tokyo 113-0033, Japan}
\author{Fumio Uchida}
\affiliation{Department of Physics, Graduate School of Science, The University of Tokyo, Bunkyo-ku, Tokyo 113-0033, Japan}
\affiliation{Research Center for the Early Universe, Graduate School of Science, The University of Tokyo, Bunkyo-ku, Tokyo 113-0033, Japan}
\author{Kotaro Fujisawa}
\affiliation{Department of Liberal Arts, Tokyo University of Technology, Ota-ku, Tokyo 144-0051, Japan}
\affiliation{Research Center for the Early Universe, Graduate School of Science, The University of Tokyo, Bunkyo-ku, Tokyo 113-0033, Japan}
\author{Naoki Yoshida}
\affiliation{Department of Physics, Graduate School of Science, The University of Tokyo, Bunkyo-ku, Tokyo 113-0033, Japan}
\affiliation{Kavli IPMU (WPI), The University of Tokyo Institutes for Advanced Study (UTIAS),
The University of Tokyo, Chiba 277-8583, Japan}

\date{\today}

\begin{abstract}
Investigating the cosmological implication of the fact that neutrino has finite mass is of importance for fundamental physics.
In particular, massive neutrino affects the formation of the large-scale structure (LSS) of the universe, and, conversely, observations of the LSS can give constraints on the neutrino mass.
Numerical simulations of the LSS formation including massive neutrino along with conventional cold dark matter is thus an important task.
For this, calculating the neutrino distribution in the phase space by solving the Vlasov equation is a suitable approach, but it requires solving the PDE in the $(6+1)$-dimensional space and is thus computationally demanding: Configuring $n_\mathrm{gr}$ grid points in each coordinate and $n_t$ time grid points leads to $O(n_\mathrm{gr}^6)$ memory space and $O(n_tn_\mathrm{gr}^6)$ queries to the coefficients in the discretized PDE.
We propose a quantum algorithm for this task.
Linearizing the Vlasov equation by neglecting the relatively weak self-gravity of the neutrino, we perform the Hamiltonian simulation to produce quantum states that encode the phase space distribution of neutrino.
We also propose a way to extract the power spectrum of the neutrino density perturbations as classical data from the quantum state by quantum amplitude estimation with accuracy $\epsilon$ and query complexity of order $\widetilde{O}((n_\mathrm{gr} + n_t)/\epsilon)$.
Our method also reduces the space complexity to $O(\mathrm{polylog}(n_\mathrm{gr}/\epsilon))$ in terms of the qubit number, while using quantum random access memories with $O(n_\mathrm{gr}^3)$ entries.
As far as we know, this is the first quantum algorithm for the LSS simulation that outputs the quantity of practical interest with guaranteed accuracy.

\end{abstract}

\maketitle

\section{Introduction}

\subsection{Background}

{\it Quantum computing} is an emerging technology and has the potential to speedup numerical tasks intractable by {\it classical computers}, today's ordinary computers including supercomputers.
Witnessing the recent rapid advance of quantum computing, people are now trying to find use cases of quantum computers with the quantum advantage in various fields.
In this paper, we focus on the simulation of the {\it large-scale structure (LSS)} of the universe with {\it massive neutrino}, an important task in cosmology.

In the standard cosmological model, all the rich structures of the present-day Universe formed through gravitational instability of tiny density fluctuations seeded in the
early universe \cite{dodelson2020modern}.
The structure at the largest scale probed by cosmological observations is called the LSS.
The evolution and the resultant LSS have been shaped by the nature of the mysterious
constituents of the universe.
Interestingly, the major components of the universe are largely unidentified. In terms of the energy fraction, about 69\% is contributed by the so-called {\it dark energy}, perhaps some unknown form of energy different from matter, and about 26\% is {\it dark matter (DM)}, which is often thought to be unknown elementary particles that are not predicted to exist
in the Standard Model (SM) of particle physics \cite{Planck2018}.
Observations of the cosmic LSS can shed light on the nature of such dark components and, eventually, provide an important clue on physics beyond the SM.

An important issue toward a better understanding of particle physics through the LSS
is that neutrinos have finite masses.
Neutrinos are massless particles in the SM,
but the detection of neutrino flavor oscillation \cite{Fukuda1998} has now established that they have nonzero mass.
The current constraint from the neutrino oscillation experiments is given as the lower bound on the neutrino mass.
Although the estimated mass is much smaller than other SM particles such as electron with $0.51\mathrm{MeV}$, the nonzero neutrino mass is definite evidence that there is physics beyond the SM.
Intriguingly, astronomical observations of the cosmic LSS provide independent constraints on the neutrino mass\footnote{We also note that there are some researches for constraining the neutrino mass through other types of particle physics experiments or astronomical observations \cite{PAGLIAROLI2010287,KATRIN:2021uub,DiValentino2021,Abe2023}.}.
Neutrinos produced in the early universe exist even today as {\it relics}, and constitute a part of matter component, along with {\it cold dark matter (CDM)}, the conventional picture of DM as particles with much larger mass, e.g., TeV order for supersymmetric particles.
In the LSS formation, neutrino behaves differently from CDM and ordinary matter (baryons), which we hereafter call CDM collectively since they behave similarly under the action of mutual gravity.
Neutrinos have non-relativistic but extremely large velocities and thus stream almost freely in the gravitational potential of CDM. Hence neutrinos used to be thought as {\it hot dark matter (HDM)} distinguishing from CDM.
The distribution and gravitational dynamics of neutrinos affect the formation of LSS differently from the conventional picture with CDM only.
This gives hope that, by comparing the theoretical prediction on the LSS with massive neutrino to the results of cosmological observations, one can derive constraints on the neutrino mass.
We also note that some particle physics models beyond the SM predict other light elementary particles such as axion \cite{Weinberg1978,Wilczek1978,Abbott:1982af,Preskill:1982cy,Dine:1982ah}, and they can behave as HDM and be constrained from LSS observations too \cite{MARSH2016}.

This background gives us a strong motivation for the numerical simulation of the LSS with massive neutrino.
However, it is a computationally demanding task.
$N$-body simulations are often chosen as a powerful method for the study of LSS, where we employ a large number of particles and follow their gravitational dynamics. 
Massive neutrinos have been also incorporated into this method, by either adopting approximate corrections or by directly represented by ``light'' particles \cite{Brandbyge_2009,Brandbyge_2010_1,Brandbyge_2010_2,Viel_2010,Bird2012,Villaescusa-Navarro_2014,Castorina_2014,Inman2015,Banerjee_2018,Villaescusa-Navarro_2020}.
However, there remains a concern that such $N$-body simulations with massive neutrino may lead to imprecise results because, unlike the conventional matter components, CDM and baryons, neutrinos have typically a much larger velocity dispersion, and the so-called shot noise can be significant unless an extremely large number of simulation particles
are employed.

An alternative approach is the Vlasov simulation, that is, solving directly the collisionless Boltzmann equation, also known as the Vlasov equation\footnote{Here, we do not explicitly take into account cosmic expansion and hereafter we consider integration over a sufficiently short time
so that the expansion parameter does not vary significantly. We leave the formulation including cosmic expansion as our future work.}:
\begin{equation}
    \frac{\partial f}{\partial t}(t,\mathbf{x},\mathbf{v})+\mathbf{v}\cdot\frac{\partial f}{\partial \mathbf{x}}(t,\mathbf{x},\mathbf{v})+\mathbf{F}(t,\mathbf{x})\cdot\frac{\partial f}{\partial \mathbf{v}}(t,\mathbf{x},\mathbf{v})=0,
    \label{eq:Vlasov}
\end{equation}
where $f(t,\mathbf{x},\mathbf{v})$ is the neutrino's distribution function in the 6D phase space consisting of the 3D position $\mathbf{x}$ and the 3D velocity $\mathbf{v}$ at time $t$, and $\mathbf{F}(t,\mathbf{x})$ is the gravitational force per unit mass on a neutrino at position $\mathbf{x}$ at time $t$.
However, solving this is also a challenging task.
Since there is no exact analytic solution of Eq.~(\ref{eq:Vlasov}) in realistic settings, we need to resort to some numerical method.
A straightforward way to solve a partial differential equation (PDE) is discretization: Setting grid points in the phase space, we can convert the PDE to a linear system of ordinary differential equations (ODEs), and then apply some algorithm for solving ODE.
This approach, though, has limited feasibility because of the so-called {\it curse of dimensionality}.
If we take $n_\mathrm{gr}$ grid points in one direction, then the total grid number in the 6D phase space becomes $n_\mathrm{gr}^6$ and so does the dimension of the ODE.
Then, in the time integration of the ODE, the memory space used is of order $O(n_\mathrm{gr}^6)$, and if we take $n_t$ time steps in the integration, then $O(n_\mathrm{gr}^6n_t)$ queries to the coefficients in the ODE are made in total.
This sixth-order scaling makes the Vlasov simulation heavier than the $N$-body simulation, where we deal with equations of motion in the 3D space. 
Although there are some studies that try to solve Eq.~(\ref{eq:Vlasov}) by supercomputers \cite{Yoshikawa_2020,Yoshikawa2021,Yoshikawa2023}, the room to increase the grid number to improve accuracy is limited.

\subsection{Our contribution}

Motivated by the above things, in this paper, we propose the quantum algorithm for the Vlasov simulation of neutrino run on a fault-tolerant quantum computer (FTQC).

Our first key observation is that we can reduce Eq.~(\ref{eq:Vlasov}) to the form to which {\it Hamiltonian simulation} \cite{Lloyd1996,berry2007efficient,childs2010relationship,Childs2011,childs2012hamiltonian,Berry2014,Berry2015Hamiltonian,Berry2015Simulating,Low2017HSimQSP,Low2019hamiltonian} can be applied.
At first glance, applying a quantum algorithm to Eq.~(\ref{eq:Vlasov}) seems difficult, since it is a nonlinear PDE.
$\mathbf{F}$ is composed of not only the gravity from CDM but also the self-gravity of neutrino and thus depends on $f$, which makes Eq.~(\ref{eq:Vlasov}) nonlinear.
Because of the unitarity of quantum operations, using quantum computing for nonlinear problems is not straightforward.
In fact, for solving differential equations, most existing quantum algorithms are dedicated to linear ones \cite{Cao_2013,Berry2014highorder,Montanaro2016,berry2017quantum,Costa2019,Xin2020,childs2020quantum,wang2020quantum,Childs2021highprecision,linden2022quantum,an2022theory,berry2022quantum,jin2022quantumsimulation,Krovi2023improvedquantum,Fang2023timemarchingbased,hu2023quantum,jin2023quantum,bagherimehrab2023fast,WEI2023494}, and those for nonlinear ones have application conditions such as the nonlinear term needing to be small in some sense \cite{lloyd2020quantum,Liu2021,Xue_2021,an2022efficient,jin2022quantum,Krovi2023improvedquantum,JIN2023112149,tanaka2023quantum}.
Nevertheless, we can approximately transform Eq. (\ref{eq:Vlasov}) into a linear PDE as follows.
Since neutrino accounts for a much smaller fraction than CDM, we can neglect the gravity from neutrino to itself and CDM unless the neutrino density is extremely non-uniform.
Then, we obtain the gravitational field $\mathbf{F}_\mathrm{CDM}(t,\mathbf{x})$ by CDM using e.g., the $N$-body simulation in advance, and approximate $\mathbf{F}(t,\mathbf{x})$ with $\mathbf{F}_\mathrm{CDM}(t,\mathbf{x})$.
The resulting linear PDE can be transformed into the ODE by discretization.
Importantly, the ODE has an antisymmetric coefficient matrix $A$ and thus is considered as a Schr\"{o}dinger equation with the Hermitian $iA$ being the Hamiltonian.
We then apply the Hamiltonian simulation, a methodology to generate a quantum state evolved under a given Hamiltonian, and yield the quantum state $\ket{\mathbf{f}(T)}$ encoding the value of $f$ on the grid points at the terminal time $T$ in the amplitudes.
As seen later, this takes the $\widetilde{O}(n_\mathrm{gr}+n_t)$\footnote{$\widetilde{O}(\cdots)$ hides the logarithmic factors in the Landau's big-O notation.} complexity in terms of the number of queries to the oracle to access the entries in $A$, which indicates a large speedup from $O(n_\mathrm{gr}^6n_t)$.

It should be noted that extracting some quantities of interest from the quantum state encoding $f$ in the amplitudes is another issue \cite{aaronson2015read}.
Then, we also present how to obtain a typical target quantity in the LSS simulation, the power spectrum $P_\nu(k)$ of the neutrino density perturbation, which indicates the magnitude of the perturbation at the specified scale $k$, from the quantum state.
This is done by some additional unitary operations on $\ket{\mathbf{f}(T)}$ and quantum amplitude estimation (QAE) \cite{brassard2002}.
In total, the proposed method outputs an estimate on the power spectrum of accuracy $\epsilon$ with $\widetilde{O}\left((n_\mathrm{gr}+n_t)/\epsilon\right)$ query complexity.

We also note that the proposed quantum algorithm can provide the advantage on space complexity, too.
In the above task for calculating $P_\nu(k)$, our algorithm uses $O(\log^{5/2}(n_\mathrm{gr}/\epsilon))$ qubits in total, which is exponentially smaller than the $O(n_\mathrm{gr}^6)$ memory space in the aforementioned classical method.
However, our method uses some quantum random access memories (QRAMs) \cite{giovannetti2008} with $O(n_\mathrm{gr}^3)$ entries to store the precomputed CDM gravitational field $\mathbf{F}_\mathrm{CDM}(t,\mathbf{x})$ on the grid points.

In addition, the proposed method can cope with the following complication in the practical problem setting.
Reflecting the stochastic nature of the initial value of the perturbation, $P_\nu(k)$ is defined as an ensemble average.
In the classical LSS simulation, it is estimated through multiple runs of the $N$-body or Vlasov simulation with different initial conditions.
In our quantum method, we do not need multiple runs: We can generate the quantum state that encodes the results from different initial conditions in superposition, and estimate $P_\nu(k)$ with the single quantum state.

\subsection{Comparison to previous studies}

Quantum algorithms for solving the Vlasov equation have been considered in previous studies.
Most of them give discussions in the context of plasma physics, and there has been no study focusing on the gravitational LSS simulation with massive neutrino as far as we know.
Also in the technical aspect, the existing studies are in directions different from ours.

Refs.~\cite{Engel2019,Ameri2023,toyoizumi2023hamiltonian} presented FTQC algorithms to solve the Vlasov--Poisson equation, in which the force induced by the particles themselves is considered, in the context of plasma physics.
They considered the linearized Vlasov--Poisson equation, which describes the perturbative solution on the zeroth-order analytical solution.
Compared to this, our approach for linearization, which approximates the force field by only that from CDM based on the small neutrino mass fraction, is different in the following points.
First, the CDM gravity is externally given by, e.g., $N$-body simulation, in which the non-linear dynamics is incorporated, and it is reflected to the neutrino dynamics solved in the current approach.
Second, the condition for our approximation to be valid is that neutrino density non-uniformity is not so large that the neutrino self-gravity is negligible compared to the CDM gravity.
This is different from the condition for the perturbative approach to be valid, the perturbation being smaller than the zeroth-order solution, which is the neutrino background density in this case.
The fractional mass (and energy) density of neutrino is less than one percent of that of CDM,
as given in Eq. (\ref{eq:NeuFrac}) later. Also, recent fully nonlinear simulations show
the maximum local over-density of neutrino is of the order unity \cite{Yu2017,Yoshikawa_2020}.
These values depend slightly on the neutrino mass but remain extremely small compared to those
of CDM whose maximum over-densities 
reach 500-1000 in nonlinear ``halos."
Hence the local gravitational potential is dominated by CDM in most
of the regime of interest.

As a difference from the other aspect, although Refs.~\cite{Engel2019,Ameri2023,toyoizumi2023hamiltonian} also used the Hamiltonian simulation, they worked in the Fourier space instead of working in the position space $\mathbf{x}$ like us.
This does not fit our setting that the position-wise CDM gravity $\mathbf{F}_\mathrm{CDM}(t,\mathbf{x})$ are given.

There are also studies on quantum algorithms to solve the Vlasov equation in the nonlinear form.
Ref.~\cite{Engel2021} considers approaches via some linearization methods such as Carleman linearization \cite{carleman1932application}.
However, like the existing quantum solvers for nonlinear differential equations \cite{lloyd2020quantum,Liu2021,Xue_2021,an2022efficient,jin2022quantum,Krovi2023improvedquantum,JIN2023112149,tanaka2023quantum}, the method in \cite{Engel2021} has some application conditions such as weakness of the nonlinearity.
If such conditions are satisfied, then quantum nonlinear differential equation solvers based on Carleman linearization provide a solution with space and query complexities of the same order as linear ones except for some logarithmic overheads \cite{Krovi2023improvedquantum}.
Ref.~\cite{dodin2021applications} summarizes the prospect of quantum algorithms to solve plasma dynamics in both linearized and non-linear settings, and also mentions the variational quantum algorithms (VQAs) \cite{cerezo2021variational}.
They might be able to run on noisy intermediate-scale quantum (NISQ) devices in the near term, but they are genuinely heuristic.

When it comes to the simulation of self-gravitating systems, in which LSS simulation is included, \cite{Mocz_2021} proposes a VQA to solve the nonlinear governing equation.
Ref.~\cite{cappelli2023vlasov} also presents a VQA, with the fuzzy DM \cite{Hu2000}, a specific scenario for DM, in mind. 
Although their numerical experiment shows a promising result in the proof-of-concept problem, it is unclear whether their methods scale to the larger problem.
\cite{TODOROVA2020109347} proposes an FTQC algorithm for the Vlasov equation based on the reservoir method in the context of fluid dynamics, and 
\cite{yamazaki2023quantum} proposes a similar algorithm for self-gravitating systems.
By their method, the Vlasov equation is simulated by appropriately arranging quantum circuits performing increment and decrement operations.
While classically controlling the arrangement of the quantum circuits, it allows the delegation of 6-dimensional operations to quantum computation.
As a result, their method can reduce the complexity scaling on $n_\mathrm{gr}$ to $O(n_\mathrm{gr}^3)$, and thus our method achieves larger speedup. 

We also comment that the method to extract the power spectrum from the quantum state has not been proposed to our knowledge. 

\subsection{Organization}

The rest of this paper is organized as follows.
In Sec.~\ref{sec:Prel}, we will explain some basics of the LSS simulation and quantum algorithms used in this paper.
Sec.~\ref{sec:OurQAlgo} is the main part, where we will explain each part of our method one by one: discretizing the Vlasov equation, obtaining the solution-encoding quantum state by Hamiltonian simulation, extracting the power spectrum from the quantum state by QAE, coping with the ensemble average, and so on.
To illustrate our algorithm, we present a demonstrative numerical experiment on the Hamiltonian simulation-based time evolution of $f(t,\mathbf{x},\mathbf{v})$ in Sec. \ref{sec:demo}.
Sec.~\ref{sec:summary} summarizes this paper.

\section{Preliminaries \label{sec:Prel}}

\subsection{Notation}

We use $I$ to denote an identity operator.
To avoid ambiguity, we may write it as $I_n$ with $n\in\mathbb{N}$ when its size is $n\times n$.
We define $\mathbb{R}_+$ by the set of all positive real numbers.
For every $n\in\mathbb{N}$, we define $[n]_0:=\{0,1,...,n-1\}$.
If a matrix $A$ has at most $s\in\mathbb{N}$ nonzero elements in each row and each column, then we say that $A$ is $s$-sparse and the sparsity of $A$ is $s$.

For a vector $\mathbf{x}\in\mathbb{C}^n$, $\|x\|$ denotes its Euclidean norm.
For an (unnormalized) quantum state $\ket{\psi}$ on a multi-qubit system, $\|\ket{\psi}\|$ denotes the Euclidean norm of its state vector.
For a matrix $A\in\mathbb{C}^{m \times n}$, $\|A\|$ denotes its spectral norm, and $\|A\|_\mathrm{max}$ denotes its max norm, the maximum of the absolute values of its entries.

For $\epsilon\in\mathbb{R}_+$, we say that $x^\prime\in\mathbb{R}$ is an $\epsilon$-approximation of $x\in\mathbb{R}$, if $|x^\prime-x|\le\epsilon$ holds.
We use $\log$ and $\lg$ for the natural and binary logarithms, respectively.

We label the position of an entry in a vector and the row and column in a matrix with an integer starting from 0. 
For example, we write a vector $\mathbf{v}\in\mathbb{C}^n$ entrywise as $\mathbf{v}=(v_0,v_1,...,v_{n-1})$ and call $v_i$ the $i$-th entry, and a matrix $A\in\mathbb{C}^{m \times n}$ as $A=(a_{ij})_{i\in[m]_0,j\in[n]_0}$ and call $a_{ij}$ the $(i,j)$-th entry. 

We denote by $\mathbf{1}_C$ the indicator function, which takes 1 if the condition $C$ is satisfied and 0 otherwise. 

\subsection{Simulation of large-scale structure formation with massive neutrinos \label{sec:ClLSSSim}}

Here, we briefly explain some basics of the LSS simulation and review recent developments of simulations that include massive neutrinos.

\subsubsection{$N$-body simulation}

$N$-body simulations are often employed to simulate the LSS formation.
In this approach, the mass distribution is represented by a collection of a large but tractable number of {\it superparticles}.
By populating the simulation volume with $N_\mathrm{p}$ superparticles, we numerically solve the time evolution equation for them:
\begin{equation}
    \frac{d^2}{dt^2}\mathbf{x}_i(t) = \mathbf{F}_i(\{\mathbf{x}_j(t)\}),
\end{equation}
where $\mathbf{x}_i(t)$ is the space coordinate of the $i$-th super particle and $\mathbf{F}_i$ is the gravitational force per unit mass on it caused by the other ones.
Then, the important statistical quantities such as the power spectrum of the density perturbation are estimated from the distribution of the superparticles.
When simulating the LSS that consists of multiple species of particles, e.g., CDM and neutrino, one may naively attempt to use
different sets of superparticles for each component.

We need to distribute the particles sufficiently densely in the 3D space in order to estimate the quantities of interest accurately.
A common setting on $N_\mathrm{p}$ is $N_\mathrm{p}=n_\mathrm{p}^3$ with large $n_\mathrm{p}$, say $O(10^2)$, which means $N_\mathrm{p}$ can be of order of million.
Thus, the $N$-body simulation is a rather heavy calculation, which may need a supercomputer in classical computing.
Nevertheless, it is tractable compared to the Vlasov simulation, which deals with the 6D space-velocity phase space.
In particular, it is commonly considered that if the velocity dispersion of the particles is small as is the case for CDM, then the $N$-body simulation yields accurate result.
On the other hand, for particles with large velocity dispersion such as massive neutrino, it is difficult to use a sufficiently large number of superparticles so that the distribution in the velocity space is represented, and often the $N$-body simulation may lead to imprecise results.

\subsubsection{Vlasov simulation for massive neutrinos}

In light of the above issue, a more desirable approach for LSS simulation with massive neutrinos is solving the Vlasov equation (\ref{eq:Vlasov}) directly.
In particular, for neutrino, we can simplify Eq.~(\ref{eq:Vlasov}) using the fact that it accounts for only a small part of the matter: The ratio of the neutrino energy fraction $\Omega_\nu$ to $\Omega_\mathrm{m}$ that of the whole of the matter is \cite{Yoshikawa_2020, Mangano:2005cc}
\begin{equation}
    \frac{\Omega_\nu}{\Omega_\mathrm{m}}=\frac{M_\nu / 93.14 \mathrm{eV}}{\Omega_\mathrm{m}h^2} \simeq 7.6 \times 10^{-3} \times \frac{M_\nu}{0.1\mathrm{eV}},
    \label{eq:NeuFrac}
\end{equation}
where $h$ is the present value of the Hubble parameter in units of $100 \mathrm{km}/ \mathrm{s}$ and we use $\Omega_\mathrm{m} h^2\simeq 0.14$ \cite{Planck2018}.
Based on this fact, we hereafter neglect the self-gravity of neutrino.
Then, Eq.~(\ref{eq:Vlasov}) approximately boils down to the linear equation
\begin{equation}
    \frac{\partial f}{\partial t}(t,\mathbf{x},\mathbf{v})+\mathbf{v}\cdot\frac{\partial f}{\partial \mathbf{x}}(t,\mathbf{x},\mathbf{v})+\mathbf{F}_\mathrm{CDM}(t,\mathbf{x})\cdot\frac{\partial f}{\partial \mathbf{v}}(t,\mathbf{x},\mathbf{v})=0,
    \label{eq:VlasovNeu}
\end{equation}
where
\begin{equation}
\mathbf{F}_\mathrm{CDM}(t,\mathbf{x})=(F_{\mathrm{CDM},x}(t,\mathbf{x}),F_{\mathrm{CDM},y}(t,\mathbf{x}),F_{\mathrm{CDM},z}(t,\mathbf{x}))
\end{equation}
is the gravitational force per unit mass on a neutrino particle at position $\mathbf{x}$ and time $t$ exerted by CDM.
This type of approximation can be found also in previous studies \cite{nascimento2023neutrino}.
We also neglect the gravitational back reaction from neutrino to CDM, and then obtain $\mathbf{F}_\mathrm{CDM}$ from the dynamics of CDM only by, e.g., the $N$-body simulation.
In other words, we consider the neutrino dynamics given the gravity by CDM as an external force field.
This leads to the difference between the equation we try to solve and those in the previous studies on the Vlasov-Poisson equation in plasma physics and LSS simulation \cite{Engel2019,Ameri2023,toyoizumi2023hamiltonian,Engel2021,dodin2021applications,Mocz_2021}.
Although both consider the Vlasov equation (\ref{eq:Vlasov}) for some particles, in the latter, the force field $\mathbf{F}$ is generated by the particles themselves and given through the Poisson equation, and thus the system of the Vlasov and Poisson equations is solved.
On the other hand, in our setting, $\mathbf{F}$ is explicitly given as a known function $\mathbf{F}_\mathrm{CDM}$, and only Eq. (\ref{eq:VlasovNeu}) is solved.

Solving Eq.~(\ref{eq:VlasovNeu}) numerically is still a difficult task.
As explained in Introduction, taking $n_\mathrm{gr}$ grid points in each of the 6 directions in the phase space leads to the total number of the grid points being $n_\mathrm{gr}^6$.
The total grid number rapidly increases with $n_\mathrm{gr}$, and so do the query and space complexity in this approach.
Although there are some studies in this direction\footnote{In \cite{Yoshikawa_2020,Yoshikawa2021,Yoshikawa2023}, taking into account the gravity by neutrino to itself and CDM, the authors performed the $N$-body simulation for CDM and the Vlasov simulation for neutrino in combination, unlike this paper neglecting the neutrino gravity.} that use supercomputers and take hundreds to thousands of grids in each position direction and tens of grids in each velocity direction \cite{Yoshikawa_2020,Yoshikawa2021,Yoshikawa2023}, increasing the grid point number largely is hard to desire in classical computing.

\subsubsection{Neutrino power spectrum}

Once we get the neutrino distribution function $f(T,\mathbf{x},\mathbf{v})$ at the terminal time $T$ of the simulation, we can derive the quantities of interest from it as follows.

A quantity we are typically interested in is the {\it neutrino power spectrum} $P_\nu(k)$.
It is defined as
\begin{equation}
    \Braket{\tilde{\delta}_\nu(T,\mathbf{k})\tilde{\delta}_\nu^*(T,\mathbf{k}^\prime)}=(2\pi)^3P_\nu(k)\delta^{(3)}(\mathbf{k}-\mathbf{k}^\prime).
\end{equation}
Here, $\tilde{\delta}_\nu(t,\mathbf{k}):=\int \delta_\nu(t,\mathbf{x})e^{-i\mathbf{k}\cdot\mathbf{x}} \mathrm{d}^3 \mathbf{x} $ is the Fourier component of the neutrino density perturbation
\begin{equation}
    \delta_\nu(t,\mathbf{x}) := \frac{\rho_\nu(t,\mathbf{x})-\bar{\rho}_\nu(t)}{\bar{\rho}_\nu(t)}
\end{equation}
at time $t$, where
\begin{equation}
    \rho_\nu(t,\mathbf{x}):=\int f(t,\mathbf{x},\mathbf{v}) \mathrm{d}^3 \mathbf{v} 
\end{equation}
is the neutrino density and $\bar{\rho}_\nu(t)$ is its average with respect to position.
$\braket{\cdot}$ denotes the ensemble average with respect to the randomness of the inflationary perturbation, which is assumed to be homogeneous.
$\delta^{(3)}(\cdot)$ denotes the 3D Dirac's delta function.
Because of the assumption that the universe is isotropic, $P_\nu$ depends on $\mathbf{k}$ only through $k:=\|\mathbf{k}\|$ the norm of the wavenumber vector.

Although neutrino and CDM behave as nonrelativistic matters in the present universe, the fluctuations of the neutrino density and the CDM density leave different imprints in cosmological observations.
It is implied that the galaxy number density traces  the mass density which is mostly contributed by CDM \cite{Villaescusa-Navarro_2014,Castorina_2014,MatteoCostanzi_2013,LoVerde2014,Castorina_2015,Villaescusa-Navarro_2018,Vagnozzi_2018,Raccanelli2018}.
The effect of neutrino is imprinted in LSS over a wide range of length scales \cite{dodelson2020modern}. 
Combining the results from, for instance, galaxy surveys and gravitational lensing observations, we can obtain the information on neutrino distribution and on the shape of $P_\nu$.
Comparing this to the result of the Vlasov simulation, we eventually obtain constraints on the neutrino mass.

\subsection{Quantum algorithms \label{sec:qalgo}}

We now introduce some basics of quantum computing and quantum algorithms used as the building blocks in this paper.
For more general basics on quantum computing, we refer to Ref. \cite{nielsen2002}.

\subsubsection{Arithmetic circuits}

In this paper, we consider computation on the system with multiple quantum registers. 
We use the fixed-point binary representation for real numbers and, for each $x\in\mathbb{R}$, we denote by $\ket{x}$ the computational basis state on a quantum register where the bit string corresponds to the binary representation of $x$.
We assume that every register has a sufficient number of qubits and thus neglect errors by finite-precision representation.
For $z=x+iy\in\mathbb{C}$, we define $\ket{z}:=\ket{x}\ket{y}$, and the real and imaginary parts are held on separate registers.

We can perform arithmetic operations on numbers represented on qubits.
For example, we can implement quantum circuits for four basic arithmetic operations: addition $O_{\mathrm{add}}:\ket{a}\ket{b}\ket{0}\mapsto\ket{a}\ket{b}\ket{a+b}$, subtraction $O_{\mathrm{sub}}:\ket{a}\ket{b}\ket{0}\mapsto\ket{a}\ket{b}\ket{a-b}$, multiplication $O_{\mathrm{mul}}:\ket{a}\ket{b}\ket{0}\mapsto\ket{a}\ket{b}\ket{ab}$, and division $O_{\mathrm{div}}:\ket{a}\ket{b}\ket{0}\ket{0}\mapsto\ket{a}\ket{b}\ket{q}\ket{r}$, where $a,b\in\mathbb{Z}$ and $q$ and $r$ are the quotient and remainder of $a/b$.
For concrete implementations, see \cite{MunosCoreas2022} and the references therein.
In the finite-precision binary representation, these operations are immediately extended to those for real numbers, and then complex numbers.
Hereafter, we collectively call these circuits arithmetic circuits.

\subsubsection{Block-encoding}

Block-encoding, which was first introduced in \cite{Low2019hamiltonian}, is a technique to encode a non-unitary matrix as an upper-left block of a unitary matrix.
Overcoming the restriction of the unitarity of quantum circuits, this block-encoding technique makes various types of matrix computation efficiently implementable, and thus is widely used in various quantum algorithms as a fundamental building block. 
We will give the rigorous definition of the block-encoding in Appendix \ref{sec:theorem}.
In some situations, there are some ways to implement the block-encoding unitary.
For example, if we have sparse-access to a matrix $A$, that is, if $A$ is sparse and we have access to oracles that provide positions and values of nonzero entries in $A$, then we can construct the block-encoding of $A$, whose implementation cost is shown as Theorem \ref{th:BlEncSp} in Appendix \ref{sec:theorem}.
The Vlasov simulation applies to this case, since, as we will see later, $A$ in this case corresponds to the sparse matrix resulting from discretizing the Vlasov equation and we can calculate its entries for given $\mathbf{F}_\mathrm{CDM}(t,\mathbf{x})$.

\subsubsection{Hamiltonian simulation}

With $H$ a Hamiltonian of a multi-qubit system and an initial state $\ket{\psi_0}$, the Hamiltonian simulation is the task to generate the quantum state $\ket{\psi_t}$ after the evolution under $H$ for time $t$.
In other words, we aim to implement the time evolution operator $\exp(-iHt)$ as a quantum circuit.
This task has been investigated as a core quantum algorithm for a long time \cite{Lloyd1996,berry2007efficient,childs2012hamiltonian}, and recently it has been shown that, given a block-encoding of $H$, we can implement a quantum circuit for the simulation of $H$ with the optimal query complexity \cite{Low2017HSimQSP,Childs2018Toward,chakraborty2019,Gilyen2019}.
Roughly speaking, we can construct a block-encoding of $\exp(-iHt)$ with $O(\|H\|t)$ queries to a block-encoding of $H$.
The rigorous statement is given as Theorem \ref{th:BlEncExpH} in Appendix \ref{sec:theorem}.

\subsubsection{Quantum amplitude estimation}

Even if we obtain some quantum state by a quantum algorithm such as Hamiltonian simulation, extracting information of interest from the state is another issue \cite{aaronson2015read}.
There is no general prescription unfortunately, and we need to devise a way to extract the quantity we want on a case-by-case basis.
In some cases, the necessary quantity is encoded in the quantum state $\ket{\Phi}$ as the amplitude of a specific basis state, and then we can use the quantum algorithm called QAE \cite{brassard2002} to estimate the amplitude.
Roughly speaking, we can estimate the amplitude with accuracy $\epsilon$, using the quantum circuit to generate $\ket{\Phi}$ $O(1/\epsilon)$ times.
The rigorous statement on the query complexity of QAE is given as Theorem \ref{th:QAE} in Appendix \ref{sec:theorem}.

\subsubsection{Quantum random access memory \label{sec:QRAM}}

As we will see later, in the algorithm we propose, we need to load the data obtained outside the quantum algorithm onto the quantum register.
A QRAM \cite{giovannetti2008} is a device for such an aim.
Specifically, we assume that, given a set of $N$ real numbers $\mathcal{X}=\{x_0,...,x_{N-1}\}$, we can implement the unitary operation on a two-register system like
\begin{equation}
    U_\mathcal{X}\sum_{i=0}^{N-1}\alpha_i\ket{i}\ket{0}=\sum_{i=0}^{N-1}\alpha_i\ket{i}\ket{x_i} \label{eq:QRAM}
\end{equation}
for any $\alpha_0,...,\alpha_{N-1}\in\mathbb{C}$ such that $\sum_{i=0}^{N-1}|\alpha_i|^2=1$.

Originally, a QRAM was proposed as a dedicated device for data loading, consisting of ``atoms" that emit ``photons" and enabling the parallel data loading in superposition in time $O(\log N)$ \cite{giovannetti2008}.
However, some issues on the implementability of such a device have been pointed out, for example, the feasibility of error correction \cite{Arunachalam_2015,jaques2023qram}.
On the other hand, there are some proposals on quantum circuit-based QRAM \cite{low2018trading,park2019circuit,Paler2020,Matteo2020,Hann2021}, into which we can incorporate error correction. 
However, proposed circuit-based implementations need $O(N)$ gates and qubits, which reduces the feasibility for large $N$, although the circuit depth can be logarithmic in $N$.
In any case, implementing a QRAM is highly challenging, and in this paper, we simply assume the availability of the operation (\ref{eq:QRAM}).

\section{Quantum algorithm for Vlasov simulation of massive neutrinos \label{sec:OurQAlgo}}

Now, we present the quantum algorithm for the neutrino Vlasov simulation, including the procedure to estimate the neutrino power spectrum.

\subsection{Discretizing the Vlasov equation \label{sec:Disc}}

Let us start from discretization, which transforms Eq.~(\ref{eq:VlasovNeu}) to a set of ODEs.

We set the range $[0,L]$ for each of the position coordinates $\mathbf{x}=(x,y,z)$ and the range $[-V,V]$ for each of the velocity coordinates $\mathbf{v}=(u,v,w)$ with sufficiently large $L,V\in\mathbb{R}_+$.
The region in the phase space we consider is thus $\mathcal{V}:=[0,L]\times[0,L]\times[0,L]\times[-V,V]\times[-V,V]\times[-V,V]$.
As a common boundary condition, we impose the periodic condition for each position coordinate, and the Dirichlet condition $f=0$ for each velocity coordinate.
Then, we introduce $n_\mathrm{gr}$ grid points in each coordinate: We set
\begin{equation}
    x_{i_x}=i_x\Delta_\mathbf{x}, \,\,\Delta_\mathbf{x}:=\frac{L}{n_\mathrm{gr}},i_x\in[n_\mathrm{gr}]_0
\end{equation}
for $x$, and similarly for $y$ and $z$, and
\begin{equation}
    u_{i_u}=-V+(i_u+1)\Delta_\mathbf{v}, \Delta_\mathbf{v}:=\frac{2V}{n_\mathrm{gr}+1},i_u\in[n_\mathrm{gr}]_0
\end{equation}
for $u$, and similarly for $v$ and $w$.
The grid points in the 6D phase space are
\begin{equation}
    (x_{i_x},y_{i_y},z_{i_z},u_{i_u},v_{i_v},w_{i_w}), i_x,i_y,i_z,i_u,i_v,i_w\in [n_\mathrm{gr}]_0.
\end{equation}
We hereafter label the grid points with vectors of indexes:
\begin{equation}
    \mathbf{i}=(i_x,i_y,i_z,i_u,i_v,i_w) \in\mathcal{I}_6:=\underbrace{[n_\mathrm{gr}]_0\times\cdots\times[n_\mathrm{gr}]_0}_6.
\end{equation}
Sometimes, we also use the label $i\in[N_\mathrm{gr}]_0$, which is related to $\mathbf{i}$ as
\begin{align}
i=\sigma(\mathbf{i})&:=i_x + i_y \times n_\mathrm{gr} + i_z \times n_\mathrm{gr}^2 \nonumber \\
&\quad +i_u \times n_\mathrm{gr}^3+i_v \times n_\mathrm{gr}^4+i_w \times n_\mathrm{gr}^5.
\label{eq:iandVeci}
\end{align}
Here, $N_\mathrm{gr}:=n_\mathrm{gr}^6$ is the total number of the grid points in the phase space.
For the later convenience, we assume that $n_\mathrm{gr}=2^{m_\mathrm{gr}}$ with some $m_\mathrm{gr}\in\mathbb{N}$.

Although we take the same number of grid points in each of position and velocity coordinates just for simplicity, we can take different numbers and such a generalization is straightforward.
In fact, in simulations in \cite{Yoshikawa_2020,Yoshikawa2021,Yoshikawa2023}, the grid number in each position coordinate is different from that in each velocity coordinate. 

Using this grid, we approximate the partial derivative with the central difference, that is,
\begin{equation}
    \frac{\partial}{\partial x} f(t,\mathbf{x}, \mathbf{v}) \simeq \frac{f(t,\mathbf{x}+\Delta_\mathbf{x}\mathbf{e}_x, \mathbf{v})-f(t,\mathbf{x}-\Delta_\mathbf{x}\mathbf{e}_x, \mathbf{v})}{2\Delta_\mathbf{x}},
    \label{eq:centDiff}
\end{equation}
where $\mathbf{e}_x=(1,0,0)$, and similarly for the derivatives with respect to $y,z,u,v$ and $w$.
Then, Eq.~(\ref{eq:VlasovNeu}) is transformed into
\begin{equation}
    \frac{d}{d t}\mathbf{f}(t)=A(t)\mathbf{f}(t).
    \label{eq:EvEq}
\end{equation}
$\mathbf{f}(t)$ is a vector in $\mathbb{R}^{N_\mathrm{gr}}$, which can be seen as $\underbrace{\mathbb{R}^{n_\mathrm{gr}}\times\cdots\times\mathbb{R}^{n_\mathrm{gr}}}_6$.
In light of this, we label each entry in $\mathbf{f}(t)$ with $\mathbf{i}\in\mathcal{I}_6$, as well as $i\in[N_\mathrm{gr}]_0$ in Eq.~(\ref{eq:iandVeci}).
Although we define $\mathbf{f}(t)$ as a solution of Eq.~(\ref{eq:EvEq}) with some initial value, we expect that the $\mathbf{i}$th entry of $\mathbf{f}(t)$ approximates the value of $f$ at time $t$ on the $\mathbf{i}$-th grid point:
\begin{equation}
    f_\mathbf{i}(t) \simeq f(t,x_{i_x},y_{i_y},z_{i_z},u_{i_u},v_{i_v},w_{i_w}).
\end{equation}
$A(t)$ is a $N_\mathrm{gr} \times N_\mathrm{gr}$ matrix and we again label its rows and columns with both $\mathbf{i}\in\mathcal{I}_6$ and $i\in[N_\mathrm{gr}]_0$.
$A(t)$ is written as
\begin{equation}
    A(t)=A_x+A_y+A_z+A_u(t)+A_v(t)+A_w(t).
    \label{eq:A}
\end{equation}
$A_x$ is defined as
\begin{equation}
    A_x = - D_{\mathrm{per}}\otimes I_{n_\mathrm{gr}} \otimes I_{n_\mathrm{gr}}\otimes E_u\otimes I_{n_\mathrm{gr}} \otimes I_{n_\mathrm{gr}},
\end{equation}
where $D_{\mathrm{per}}$ and $E_u$ are $n_\mathrm{gr}\times n_\mathrm{gr}$ matrices defined as
\begin{equation}
    D_{\mathrm{per}} := 
    \begin{pmatrix}
    0  & 1 &   &  &  -1 \\    
    -1 & 0 & 1 &  & \\
    & \ddots & \ddots & \ddots & \\
    & & -1 & 0 & 1 \\
    1 & &  & -1& 0
    \end{pmatrix}
\end{equation}
and
\begin{equation}
    E_u :=
    \begin{pmatrix}
    \frac{u_0}{2\Delta_\mathbf{x}}  &  & \\    
    & \ddots & \\
    & & \frac{u_{n_\mathrm{gr}-1}}{2\Delta_\mathbf{x}}
    \end{pmatrix},
\end{equation}
respectively.
$A_y$ and $A_z$ are defined similarly.
$A_u(t)$ is a time-varying $N_\mathrm{gr} \times N_\mathrm{gr}$ matrix and its $(\mathbf{i},\mathbf{j})$ entry is 
\begin{align}
    &(A_u(t))_{\mathbf{i},\mathbf{j}} = \nonumber \\
    \quad &
    \begin{cases}
        -\frac{F_{\mathrm{CDM},x}(t,x_{i_x},y_{i_y},z_{i_z})}{2\Delta_\mathbf{v}} & ; \ \mathrm{if} \ \mathbf{j} = \mathbf{i} + \mathbf{e}_u \\
        \frac{F_{\mathrm{CDM},x}(t,x_{i_x},y_{i_y},z_{i_z})}{2\Delta_\mathbf{v}} & ; \ \mathrm{if} \ \mathbf{j} = \mathbf{i} - \mathbf{e}_u \\
        0 & ; \ \mathrm{otherwise}
    \end{cases},
\end{align}
where $\mathbf{e}_u:=(0,0,0,1,0,0)$.
$A_v(t)$ and $A_w(t)$ are defined similarly.

We now comment on what are conserved in the time evolution of the ODE system (\ref{eq:EvEq}).
First, $\|\mathbf{f}(t)\|$ the norm of the solution is constant in time, which is seen as follows.
$A$ is antisymmetric, as is easily checked from the definition.
Thus,
\begin{equation}
H(t):=iA(t)
\label{eq:H}
\end{equation}
is Hermitian, and we can regard Eq.~(\ref{eq:EvEq}) as the Schr\"{o}dinger equation with the Hamiltonian $H$.
As a property of the Schr\"{o}dinger equation, $\|\mathbf{f}(t)\|$ is conserved.
This means that in the ODE system (\ref{eq:EvEq}), there is no instability such that the solution explodes.
This is a merit of the current discretization as Eq.~(\ref{eq:centDiff}), which leads to the antisymmetricity of $A$.

Second, $f_\mathrm{sum}(t):=\sum_{\mathbf{i}\in\mathcal{I}_6}f_\mathbf{i}(t)$, the sum of all the entries in $\mathbf{f}(t)$ is approximately conserved.
This is seen by noting that
\begin{align}
    &\frac{d}{dt}f_\mathrm{sum}(t) \nonumber \\
    =& \sum_{\mathbf{i},\mathbf{j}\in\mathcal{I}_6} (A(t))_{\mathbf{i},\mathbf{j}} f_\mathbf{\mathbf{j}}(t)  \nonumber \\
    =& \sum_{\mathbf{i}\in\mathcal{B}^{\mathrm{l}}_u} (A(t))_{\mathbf{i}+\mathbf{e}_u,\mathbf{i}} f_\mathbf{\mathbf{i}}(t) + \sum_{\mathbf{i}\in\mathcal{B}^{\mathrm{u}}_u} (A(t))_{\mathbf{i}-\mathbf{e}_u,\mathbf{i}} f_\mathbf{\mathbf{i}}(t) \nonumber \\
    & + (u \leftrightarrow v) + (u \leftrightarrow w),
\end{align}
where $(u \leftrightarrow v)$ (resp. $(u \leftrightarrow w)$) denotes the same as the first and second terms expect $u$ is replaced with $v$ (resp. $w$), and $\mathcal{B}^{\mathrm{l}}_u:=\left\{\mathbf{i}\in\mathcal{I}_6 \ \middle| \ i_u=0 \right\}$, $\mathcal{B}^{\mathrm{u}}_u:=\left\{\mathbf{i}\in\mathcal{I}_6 \ \middle| \ i_u=n_\mathrm{gr}-1 \right\}$, and similarly defined $\mathcal{B}^{\mathrm{l}}_v,\mathcal{B}^{\mathrm{u}}_v,\mathcal{B}^{\mathrm{l}}_w,\mathcal{B}^{\mathrm{u}}_w$ are the sets of the indexes of the grid points on the boundaries in the velocity coordinates.
That is, the time derivative of $f_\mathrm{sum}$ is almost vanishing except for the contributions from the boundary grids.
The conservation of $f_\mathrm{sum}$ has the physical meaning that the number of the neutrino particles in the 6D simulation box is conserved, which is desirable in the simulation.
The violation of the conservation corresponds to the escapes of the particles from the box, which is controllable by setting the box sufficiently large.

Although these conservation properties are desirable, there might be other properties the solution $\mathbf{f}(t)$ should have from the physical perspective.
For example, since $\mathbf{f}(t)$ denotes the distribution density function in the phase space, it must be nonnegative, but, it is not clear that this is satisfied in the current approach. 
We leave incorporating some schemes guaranteeing the nonnegativity such as upwind difference into the quantum algorithm as future work.

\subsection{Generating the solution-encoding quantum state \label{sec:GenQSHamSim}}

As seen above, we can regard Eq.~(\ref{eq:EvEq}) as the Schr\"{o}dinger equation.
If we rewrite it {\it a la} quantum mechanics, then it becomes
\begin{equation}
    \frac{d}{d t}\ket{\mathbf{f}(t)}=-iH(t)\ket{\mathbf{f}(t)},
    \label{eq:ShEq}
\end{equation}
where $\ket{\mathbf{f}(t)}$ is the quantum state encoding the value of $\mathbf{f}(t)$ in the amplitudes:
\begin{equation}
\ket{\mathbf{f}(t)}:=\frac{1}{\|\mathbf{f}(t)\|}\sum_{\mathbf{i}\in\mathcal{I}_6} f_\mathbf{i}(t) \ket{\mathbf{i}}.
\label{eq:ketft}
\end{equation}
Here, $\ket{\mathbf{i}}:=\ket{i_x}\ket{i_y}\ket{i_z}\ket{i_u}\ket{i_v}\ket{i_w}$, which can be also seen as $\ket{i}$ under the correspondence between $\mathbf{i}$ and $i$ in Eq.~(\ref{eq:iandVeci}), since concatenating the binary representations of $i_x,...,i_w$ yields that of $i$.

Now, let us consider generating the quantum state $\ket{\mathbf{f}(T)}$ encoding $f$ at the terminal time $T$ by Hamiltonian simulation.
To apply this, we make the following assumption.

\begin{assumption}[Piecewise time-constancy of $\mathbf{F}_\mathrm{CDM}$ and oracles to access it]
Let $T\in\mathbb{R}_+$ and $n_t\in\mathbb{N}$.
Let $t_{i_t}=i_t \Delta_t$ for $i_t\in[n_t]_0$ with $\Delta_t:=T/n_t$. 
Then, for any $\mathbf{i}_\mathbf{x}:=(i_x,i_y,i_z)\in\mathcal{I}_3:=[n_\mathrm{gr}]_0 \times [n_\mathrm{gr}]_0 \times [n_\mathrm{gr}]_0$ and $i_t\in[n_t]_0$,
\begin{align}
    F_{\mathrm{CDM},x}(t,\mathbf{x}_{\mathbf{i}_\mathbf{x}}) &=
    F_{\mathrm{CDM},x}^{i_t,\mathbf{i}_\mathbf{x}} \nonumber \\
    F_{\mathrm{CDM},y}(t,\mathbf{x}_{\mathbf{i}_\mathbf{x}}) &=
    F_{\mathrm{CDM},y}^{i_t,\mathbf{i}_\mathbf{x}} \nonumber \\
    F_{\mathrm{CDM},z}(t,\mathbf{x}_{\mathbf{i}_\mathbf{x}}) &=
    F_{\mathrm{CDM},z}^{i_t,\mathbf{i}_\mathbf{x}}
    \label{eq:FCDMConst}
\end{align}
holds for any $t\in[t_{i_t},t_{i_t+1})$.
Here, $\mathbf{x}_{\mathbf{i}_\mathbf{x}}:=(x_{i_x},y_{i_y},z_{i_z})$, and, for each $(i_t,\mathbf{i}_\mathbf{x})\in[n_t]_0 \times \mathcal{I}_3$, $F_{\mathrm{CDM},x}^{i_t,\mathbf{i}_\mathbf{x}},F_{\mathrm{CDM},y}^{i_t,\mathbf{i}_\mathbf{x}}$ and $F_{\mathrm{CDM},z}^{i_t,\mathbf{i}_\mathbf{x}}$ are some real numbers.
Furthermore, for any $i_t\in[n_t]_0$, we are given accesses to the oracles $O^{i_t}_{F_{\mathrm{CDM},x}}$, $O^{i_t}_{F_{\mathrm{CDM},y}}$, and $O^{i_t}_{F_{\mathrm{CDM},z}}$ that act as
\begin{align}
    O^{i_t}_{F_{\mathrm{CDM},x}} \ket{\mathbf{i}_\mathbf{x}}\ket{0} &= \ket{\mathbf{i}_\mathbf{x}}\ket{F_{\mathrm{CDM},x}^{i_t,\mathbf{i}_\mathbf{x}}} \nonumber \\
    O^{i_t}_{F_{\mathrm{CDM},y}} \ket{\mathbf{i}_\mathbf{x}}\ket{0} &= \ket{\mathbf{i}_\mathbf{x}}\ket{F_{\mathrm{CDM},y}^{i_t,\mathbf{i}_\mathbf{x}}} \nonumber \\
    O^{i_t}_{F_{\mathrm{CDM},z}} \ket{\mathbf{i}_\mathbf{x}}\ket{0} &= \ket{\mathbf{i}_\mathbf{x}}\ket{F_{\mathrm{CDM},z}^{i_t,\mathbf{i}_\mathbf{x}}}
    \label{eq:OracleFCDM}
\end{align}
for any $\mathbf{i}_\mathbf{x}\in\mathcal{I}_3$, where $\ket{\mathbf{i}_\mathbf{x}}:=\ket{i_x}\ket{i_y}\ket{i_z}$.
\label{ass:FCDM}
\end{assumption}

The assumption that $\mathbf{F}_\mathrm{CDM}$ is piecewise constant in time matches the practical setting.
From the $N$-body simulation, we obtain $\mathbf{F}_\mathrm{CDM}$ only on the discrete time points, since in the $N$-body simulation we discretize the time for time integration of the equation of motion, and interpolating them as Eq.~(\ref{eq:FCDMConst}) is a common approximation.
We also note that, for this type of $\mathbf{F}_\mathrm{CDM}$, the Hamiltonian $H(t)$ is also piecewise constant in time: for any $i_t\in[n_t]_0$,
\begin{equation}
    H(t) = H_{i_t}:=H(t_{i_t})
    \label{eq:HPieceConst}
\end{equation}
holds for any $t\in[t_{i_t},t_{i_t+1})$.

We also assume that we can prepare the quantum state encoding the initial value of $\mathbf{f}(t)$.

\begin{assumption}
    For a given initial value $\mathbf{f}(0)$, we have an access to the oracle $O_{\mathbf{f}(0)}$ that acts as
    \begin{equation}
        O_{\mathbf{f}(0)}\ket{0}=\ket{\mathbf{f}(0)}.
        \label{eq:Of0}
    \end{equation}
    \label{ass:IniStOracle}
\end{assumption}

We will discuss the implementation of the above oracles in Sec.~\ref{sec:oracle}.

Then, we can apply Hamiltonian simulation to generate $\ket{\mathbf{f}(T)}$.

\begin{theorem}
    Suppose that Assumptions \ref{ass:FCDM} and \ref{ass:IniStOracle} hold.
    Then, for any $\epsilon\in(0,1/2)$, we have a quantum circuit $U_{\mathbf{f}(T),\epsilon}$ on the $(2\lg N_\mathrm{gr}+5)$-qubit system that acts as
    \begin{equation}
        U_{\mathbf{f}(T),\epsilon}\ket{0}^{\otimes(2\lg N_\mathrm{gr}+5)}=\ket{0}^{\otimes(\lg N_\mathrm{gr}+5)}\ket{{\mathbf{f}(T)}} + \ket{\psi_\mathrm{gar}},
        \label{eq:UfT}
    \end{equation}
    where $\ket{\mathbf{f}(T)}$ is the quantum state that encodes the solution of Eq.~(\ref{eq:EvEq}) at time $T$ with the initial value $\mathbf{f}(0)$ as Eq.~(\ref{eq:ketft}), and $\ket{\psi_\mathrm{gar}}$ is some unnormalized quantum state satisfying $\|\ket{\psi_\mathrm{gar}}\|\le\epsilon$.
    In $U_{\mathbf{f}(T),\epsilon}$, (controlled) $\left\{O^{i_t}_{F_{\mathrm{CDM},x}}, O^{i_t}_{F_{\mathrm{CDM},y}}, O^{i_t}_{F_{\mathrm{CDM},z}}\right\}_{i_t}$, arithmetic circuits, and their inverses are queried
    \begin{equation}
        O\left(n_\mathrm{gr}T\times\max\left\{\frac{V}{L},\frac{F_\mathrm{max}}{V}\right\} + n_t\log\left(\frac{n_t}{\epsilon}\right) \right)
        \label{eq:CompfTGen}
    \end{equation}
    times, $O_{\mathbf{f}(0)}$ is queried once, and
    \begin{equation}
        O\left(\log^{5/2}\left(\frac{n_\mathrm{gr}T}{\epsilon}\times\max\left\{\frac{V}{L},\frac{F_\mathrm{max}}{V}\right\} \right) \right)
        \label{eq:QubitfTGen}
    \end{equation}
    qubits are used including ancillary ones.
    \label{th:ketfTGen}
\end{theorem}

The proof of this theorem is given in Appendix \ref{sec:PrThKetfT}.
Although we postpone the details of construction of $U_{\mathbf{f}(T),\epsilon}$, we now give a rough outline.
For the piecewise constant $H(t)$, the solution of Eq. (\ref{eq:ShEq}) is
\begin{equation}
    \ket{\mathbf{f}(T)}=\exp\left(-i \Delta_t H_{n_t-1}\right) \cdots \exp\left(-i \Delta_t H_{0}\right)\ket{\mathbf{f}(0)}.
    \label{eq:fTExp}
\end{equation}
We generate this state via implementing $\exp\left(-i \Delta_t H_{i_t}\right)$ by block-encoding, which is possible with the oracles $O^{i_t}_{F_{\mathrm{CDM},x}}$, etc.

We can simplify the query complexity bound (\ref{eq:CompfTGen}) by taking into account how to set $L$ and $V$ in practice.
To solve the Vlasov equation precisely, we set the considered region $\mathcal{V}$ in the phase space so that it contains the region where the distribution function $f$ takes a nonnegligible value.
In other words, noting that $f$ reflects the neutrino motion, we set $\mathcal{V}$ so that no particle escapes from it.
For this purpose, it is sufficient to set $L$ and $V$, the side lengths of $\mathcal{V}$, as
\begin{equation}
    VT \sim L,F_\mathrm{max}T \sim V,
    \label{eq:LVSuff}
\end{equation}
since in the simulation time $T$ a particle traverses at most $VT$ in the position coordinate and at most $F_\mathrm{max}T$ in the velocity coordinate.
Conversely, much larger values of $L$ and $V$ than these are too much.
Thus, assuming Eq.~(\ref{eq:LVSuff}), we can simplify Eq. (\ref{eq:CompfTGen}) as
\begin{equation}
    \widetilde{O}(n_\mathrm{gr}+n_t),
\end{equation}
as announced in Introduction.
Besides, Eq. (\ref{eq:QubitfTGen}) becomes
\begin{equation}
    O\left(\log^{5/2}\left(\frac{n_\mathrm{gr}}{\epsilon}\right)\right),
    \label{eq:qubitNumSimple}
\end{equation}
which means that the number of qubits used is only logarithmic.

\subsection{Extracting the power spectra \label{sec:Extract}}

\begin{figure*}[tp]
    \centering
    \includegraphics[width=1\linewidth]{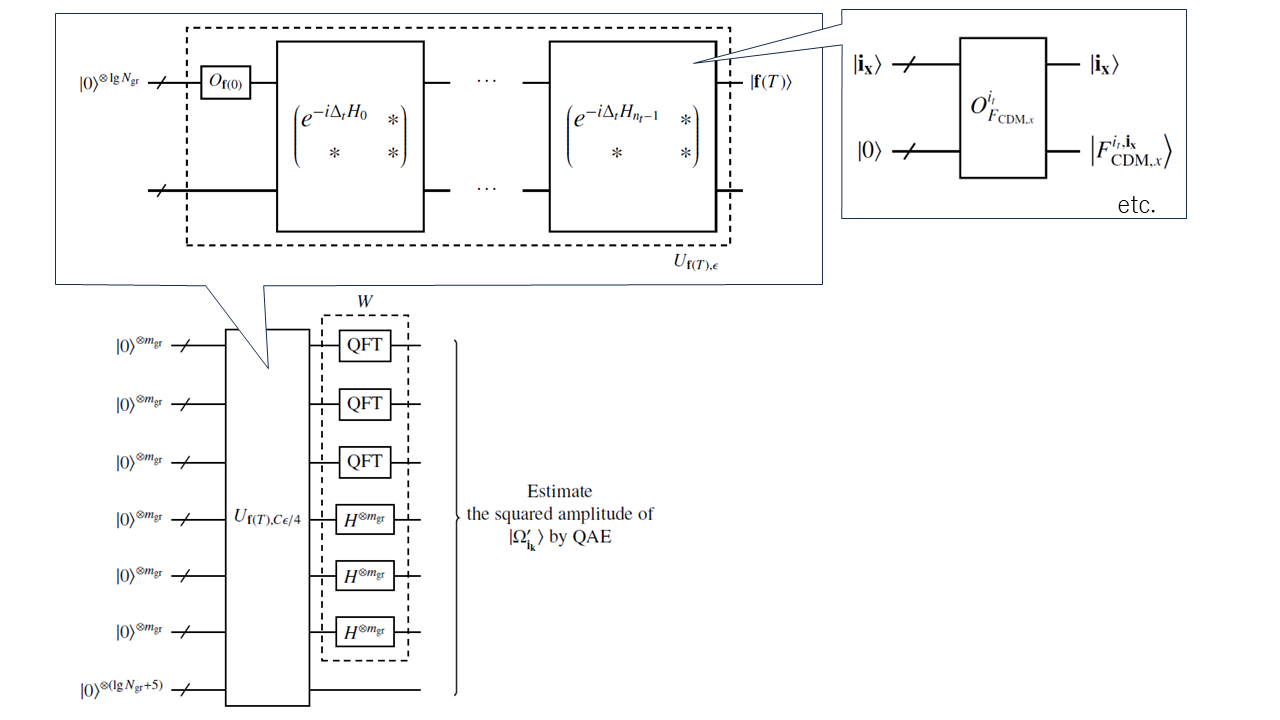}
    \caption{Overview of the quantum circuit used in Algorithm \ref{alg:PnuEstim}.}
    \label{fig:CircuitOverview}
\end{figure*}

Even though we can generate the quantum state $\ket{\mathbf{f}(T)}$ that encodes $\mathbf{f}(T)$, extracting some information of interest from it is another issue.
Here, we consider how to estimate the neutrino power spectrum from the quantum state.
In the current setting of finite volume and discrete grid points, the power spectrum is given as \cite{Jing_2005}
\begin{equation}
    P_\nu(\mathbf{k}_{\mathbf{i}_\mathbf{k}})=\left\langle \left| \tilde{\delta}^\nu_{\mathbf{i}_\mathbf{k}} \right|^2 \right\rangle
    \label{eq:PnuDisc}
\end{equation}
for $\mathbf{k}_{\mathbf{i}_\mathbf{k}}=\frac{2\pi}{L}\mathbf{i}_\mathbf{k}$ with $\mathbf{i}_\mathbf{k}\in\mathcal{I}_3$.
Here,
\begin{align}
    \tilde{\delta}^\nu_{\mathbf{i}_\mathbf{k}}
    &:=\frac{1}{n_\mathrm{gr}^3}\sum_{\mathbf{i}_\mathbf{x}\in\mathcal{I}_3} \delta^\nu_{\mathbf{i}_\mathbf{x}} \exp(i\mathbf{k}_{\mathbf{i}_\mathbf{k}}\cdot\mathbf{x}_{\mathbf{i}_\mathbf{x}})  \nonumber \\
    &= \frac{1}{n_\mathrm{gr}^3}\sum_{\mathbf{i}_\mathbf{x}\in\mathcal{I}_3}  \delta^\nu_{\mathbf{i}_\mathbf{x}} \exp\left(\frac{2 \pi i \mathbf{i}_\mathbf{k}\cdot  \mathbf{i}_\mathbf{x}}{n_\mathrm{gr}}\right)
    \label{eq:tilDelNu}
\end{align}
is the discrete Fourier transform of
\begin{equation}
    \delta^\nu_{\mathbf{i}_\mathbf{x}}:=\frac{\rho^\nu_{\mathbf{i}_\mathbf{x}}-\bar{\rho}_\nu(T)}{\bar{\rho}_\nu(T)},
    \label{eq:delNu}
\end{equation}
the neutrino density perturbation at $\mathbf{x}_{\mathbf{i}_\mathbf{x}}$ the grid point in the position space.
\begin{equation}
    \rho^\nu_{\mathbf{i}_\mathbf{x}}:=\sum_{\mathbf{i}_\mathbf{v}\in\mathcal{I}_3} f_{(\mathbf{i}_\mathbf{x},\mathbf{i}_\mathbf{v})}(T) \Delta_\mathbf{v}^3
    \label{eq:rhonud}
\end{equation}
is the neutrino density at the position space grid point, calculated by integrating $f$ with respect to the velocity coordinates in the phase space.
In Eq.~(\ref{eq:rhonud}), $(\mathbf{i}_\mathbf{x},\mathbf{i}_\mathbf{v})\in\mathcal{I}_6$ is made by concatenating $\mathbf{i}_\mathbf{x},\mathbf{i}_\mathbf{v}\in\mathcal{I}_3$.

Leaving how to take the ensemble average in Eq.~(\ref{eq:PnuDisc}) to Sec.~\ref{sec:ext}, we now consider how to estimate $\left| \tilde{\delta}^\nu_{\mathbf{i}_\mathbf{k}} \right|^2$ from $\ket{\mathbf{f}(T)}$.
Formally, we have the following theorem.

\begin{theorem}
    Suppose that Assumptions \ref{ass:FCDM} and \ref{ass:IniStOracle} hold.
    Suppose that we are given the value of
    \begin{equation}
    C:=\frac{\frac{1}{N_\mathrm{gr}}\left(f_\mathrm{sum}(T)\right)^2}{\|\mathbf{f}(T)\|^2}.
    \end{equation}
    Let $\epsilon\in(0,1/2)$, $\delta\in(0,1)$, and $\mathbf{i}_\mathbf{k}\in\mathcal{I}_3\setminus\{(0,0,0)\}$.
    Then, there exists a quantum algorithm that, with probability at least $1-\delta$, outputs an $\epsilon$-approximation of $\left| \tilde{\delta}^\nu_{\mathbf{i}_\mathbf{k}} \right|^2$ with $\mathbf{f}(T)$ the solution of Eq.~(\ref{eq:EvEq}).
    In this algorithm, (controlled) $\left\{O^{i_t}_{F_{\mathrm{CDM},x}}, O^{i_t}_{F_{\mathrm{CDM},y}}, O^{i_t}_{F_{\mathrm{CDM},z}}\right\}_{i_t}$, arithmetic circuits, and their inverses are queried
    \begin{align}
        & O\left(\left(n_\mathrm{gr}T\times\max\left\{\frac{V}{L},\frac{F_\mathrm{max}}{V}\right\} + n_t\log\left(\frac{n_t}{C\epsilon}\right)\right) \right. \nonumber \\
         & \left.\qquad \times \frac{1}{C\epsilon}\log\left(\frac{1}{\delta}\right) \right)
        \label{eq:CompfPnuEstim1}
    \end{align}
    times, $O_{\mathbf{f}(0)}$ is queried
    \begin{equation}
    O\left(\frac{1}{C\epsilon}\log\left(\frac{1}{\delta}\right) \right)
    \label{eq:CompfPnuEstim2}
    \end{equation}
    times, and
    \begin{equation}
        O\left(\log^{5/2}\left(\frac{n_\mathrm{gr}T}{C\epsilon}\times\max\left\{\frac{V}{L},\frac{F_\mathrm{max}}{V}\right\} \right) \right)
        \label{eq:QubitPnuEstim1}
    \end{equation}
    qubits are used.
    \label{th:PnuEstim}
\end{theorem}

\begin{proof}
    Supposing that we are given $\ket{\mathbf{f}(T)}=\frac{1}{\|\mathbf{f}(T)\|}\sum_{\mathbf{i}\in\mathcal{I}_6} f_\mathbf{i}(T) \ket{i_x}\ket{i_y}\ket{i_z}\ket{i_u}\ket{i_v}\ket{i_w}$ and regarding this as the quantum state on the system of six $m_\mathrm{gr}$-qubit registers, we consider the following unitary on this system:
    \begin{equation}
        W:=Q_{m_\mathrm{gr}} \otimes Q_{m_\mathrm{gr}} \otimes Q_{m_\mathrm{gr}} \otimes H_{m_\mathrm{gr}} \otimes H_{m_\mathrm{gr}} \otimes H_{m_\mathrm{gr}}.
        \label{eq:W}
    \end{equation}
    Here, $Q_{m_\mathrm{gr}}$ is the quantum Fourier transform on a $m_\mathrm{gr}$-qubit system, which acts as
    \begin{equation}
        Q_{m_\mathrm{gr}}\ket{j}=\frac{1}{\sqrt{n_\mathrm{gr}}}\sum_{l=0}^{n_\mathrm{gr}-1} \exp\left(\frac{2 \pi i jl}{n_\mathrm{gr}}\right)\ket{l}
    \end{equation}
    for any $j\in[n_\mathrm{gr}]_0$, and $H_{m_\mathrm{gr}}$ is the operation on the same system that is applying the Hadamard gate to each qubit.
    With $m_\mathrm{gr}=\lg n_\mathrm{gr}$, $Q_{m_\mathrm{gr}}$ is implemented as a circuit consisting of $O(\log^2 n_\mathrm{gr})$ Hadamard gates and conditional rotation gates with depth of the same order \cite{nielsen2002}, and $H_{m_\mathrm{gr}}$ is just a collection of $O(\log n_\mathrm{gr})$ Hadamard gates.
    We thus neglect their costs in the following discussion.
    We also note that the average density $\bar{\rho}_\nu(T)$ is related to $\mathbf{f}(T)$ as
    \begin{equation}
        \bar{\rho}_\nu(T)\Delta_\mathbf{x}^3=\frac{1}{N_\mathrm{gr}}\sum_{\mathbf{i}\in\mathcal{I}_6}f_\mathbf{i}(T).
    \end{equation}
    Using these along with Eqs.~(\ref{eq:tilDelNu}), (\ref{eq:delNu}), and (\ref{eq:rhonud}), we obtain by a straightforward calculation that, for $\mathbf{i}_\mathbf{k}:=(i_{k_x},i_{k_y},i_{k_z})\in\mathcal{I}_3$,
    \begin{equation}
        \left|\bra{\Omega_{\mathbf{i}_\mathbf{k}}}W\ket{\mathbf{f}(T)}\right|^2=C\left|\tilde{\delta}^\nu_{\mathbf{i}_\mathbf{k}}\right|^2,
    \end{equation}
    where $\ket{\Omega_{\mathbf{i}_\mathbf{k}}}:=\ket{i_{k_x}}\ket{i_{k_y}}\ket{i_{k_z}}\ket{0}\ket{0}\ket{0}$.
    Since the left-hand side is the squared amplitude of the computational basis state $\ket{\Omega_{\mathbf{i}_\mathbf{k}}}$ in $W\ket{\mathbf{f}(T)}$, we can estimate this by QAE\footnote{Note that for $\ket{\Omega_{\mathbf{i}_\mathbf{k}}}$ the operator $V$ in Theorem \ref{th:QAE}, which flips the amplitude only for $\ket{\Omega_{\mathbf{i}_\mathbf{k}}}$ and does nothing otherwise, can be implemented by X gates and a multi-controlled Z gate.}.

    In summary, we can get an $\epsilon$-approximation of $\left|\tilde{\delta}^\nu_{\mathbf{i}_\mathbf{k}}\right|^2$ by the procedure shown in Algorithm \ref{alg:PnuEstim}.

    \begin{algorithm}[H]
    \begin{algorithmic}[1]
    \REQUIRE Accuracy $\epsilon\in(0,1/2)$, success probability $1-\delta\in(0,1)$, the value of $C$.

    \STATE Construct the unitary quantum circuit $U_{\mathbf{f}(T),C\epsilon/4}$ on a $(2\lg N_\mathrm{gr}+5)$-qubit system following Theorem \ref{th:ketfTGen}.

    \STATE Construct $W^\prime:=I_{32N_\mathrm{gr}} \otimes W$ the unitary on the same system, where $I_{32N_\mathrm{gr}}$ is the identity operator on the first $\lg N_\mathrm{gr}+5$ qubits and $W$ is the operator in Eq.~(\ref{eq:W}) that acts on the other $\lg N_\mathrm{gr}$ qubits.

    \STATE Estimate the squared amplitude of $\ket{\Omega^\prime_{\mathbf{i}_\mathbf{k}}}:=\ket{0}^{\otimes(\lg N_\mathrm{gr}+5)}\ket{\Omega_{\mathbf{i}_\mathbf{k}}}$ in the state $W^\prime U_{\mathbf{f}(T),C\epsilon/4} \ket{0}$ by QAE with accuracy $C\epsilon/4$ and success probability $1-\delta$, and let the estimate be $\tilde{p}$.

    \STATE Output $\tilde{p}/C$.
    
    \caption{Estimation of $\left| \tilde{\delta}^\nu_{\mathbf{i}_\mathbf{k}} \right|^2$}
    \label{alg:PnuEstim}
    \end{algorithmic}
\end{algorithm}

The rest of the proof is on the accuracy and complexity of Algorithm \ref{alg:PnuEstim}.
We postpone it to Appendix \ref{sec:PrThPnuEstim}.
\end{proof}

To illustrate the outline of Algorithm \ref{alg:PnuEstim}, we present the overview diagram in FIG. \ref{fig:CircuitOverview}.
Using the oracles $O^{i_t}_{F_\mathrm{CDM},x}$ etc. that give the CDM gravity, we construct block-encodings of the time evolution operators $\exp(-i \Delta_t H_{i_t})$.
Combining these block-encodings and the initial-state-encoding oracle $O_{\mathbf{f}(0)}$, we construct the unitary $U_{\mathbf{f}(T),\epsilon}$ to generate the state $\ket{\mathbf{f}(T)}$ encoding the solution $\mathbf{f}(T)$ at time $T$ with accuracy $\epsilon$.
Then, this unitary followed by $W$ generate the quantum state, in which the squared amplitude of a specific computational basis state $\ket{\Omega^\prime_{\mathbf{i}_\mathbf{k}}}$ is equal to $C\left|\tilde{\delta}^\nu_{\mathbf{i}_\mathbf{k}}\right|^2$.
Thus, we estimate this squared amplitude by QAE to get $\left|\tilde{\delta}^\nu_{\mathbf{i}_\mathbf{k}}\right|^2$, the squared amplitude of the specified Fourier mode of the neutrino density perturbation.
Note that, in this QAE, we use $U_{\mathbf{f}(T),C\epsilon/4}$, whose accuracy is $C\epsilon/4$, to guarantee the accuracy $\epsilon$ for the estimation of $\left|\tilde{\delta}^\nu_{\mathbf{i}_\mathbf{k}}\right|^2$.

Let us make some comments on $C$.
In Theorem \ref{th:PnuEstim}, we assume that we know the value of $C$ in advance, even though it is defined with $\mathbf{f}(T)$ the solution at the terminal time.
In fact, this assumption is plausible.
This is because, as explained above, $\|\mathbf{f}(t)\|$ is constant over time, and $f_\mathrm{sum}(t)$ is also almost constant.
After all, we can calculate $C$ using the initial value $\mathbf{f}(0)$ instead of $\mathbf{f}(T)$.
The cost for this preliminary calculation will be discussed in Sec.~\ref{sec:oracle}.

Also, note that $C$ is written as
\begin{equation}
    C=\frac{\left(\frac{1}{N_\mathrm{gr}}\sum_{\mathbf{i}\in\mathcal{I}_6}f_\mathbf{i}(0)\right)^2}{\frac{1}{N_\mathrm{gr}}\sum_{\mathbf{i}\in\mathcal{I}_6}\left(f_\mathbf{i}(0)\right)^2},
    \label{eq:C2}
\end{equation}
that is, the ratio of the squared average of $f_\mathbf{i}(T)$ over the grid points to the average of $f_\mathbf{i}(T)$ squared, which we expect is of order 1.

By using Eq.~(\ref{eq:LVSuff}) and $C\sim1$, we simplify Eq.~(\ref{eq:CompfPnuEstim1}) to
\begin{equation}
    \widetilde{O}\left(\frac{n_\mathrm{gr}+n_t}{\epsilon}\right),
    \label{eq:CompAppPnuEstim}
\end{equation}
which is the query complexity bound announced in Introduction.
Besides, Eq. (\ref{eq:QubitPnuEstim1}) becomes of order (\ref{eq:qubitNumSimple}), the space complexity announced in the Introduction.

\subsection{Extensions \label{sec:ext}}

Equipped with Algorithm \ref{alg:PnuEstim} as a basic one, we now consider its extensions so that it matches complications in practice.

\subsubsection{Integrated power spectrum}

Although Algorithm \ref{alg:PnuEstim} outputs an estimate of $\left| \tilde{\delta}^\nu_{\mathbf{i}_\mathbf{k}} \right|^2$, this might not be feasible, since the magnitude of $\left| \tilde{\delta}^\nu_{\mathbf{i}_\mathbf{k}} \right|^2$ may be tiny.
In fact, in the current discrete setting with $n_\mathrm{gr}^3$ grid points in the position coordinates, $\delta_\nu(T,\mathbf{x})$ consists of the $n_\mathrm{gr}^3$ Fourier components, and thus $\left| \tilde{\delta}^\nu_{\mathbf{i}_\mathbf{k}} \right|^2$ the amplitude of each Fourier components is typically suppressed by a factor $1/n_\mathrm{gr}^3$.
Therefore, to obtain a nonzero estimate of a value of such an order, we need to run QAE with $O(1/n_\mathrm{gr}^3)$ accuracy and $O(n_\mathrm{gr}^3)$ query complexity.
After all, the query complexity scaling like Eq.~(\ref{eq:CompAppPnuEstim}) is not achieved.

However, this issue is not serious in practice.
First, under the usual assumption of isotropy, $P_\nu(\mathbf{k})$ does not depend on the direction of $\mathbf{k}$ but only its norm $k=|\mathbf{k}|$.
Furthermore, for the purpose of comparing the numerical simulation with observations, it often suffices to get the integrated power spectrum
\begin{equation}
    \bar{P}_\nu(k_1,k_2):=\int_{k_1}^{k_2} dk 4\pi k^{2} P_\nu(k)
\end{equation}
with some interval $[k_1,k_2]$, which indicates the total magnitude of the neutrino density perturbations of scales between $k_1$ and $k_2$.
In the current discrete setting, the corresponding quantity is
\begin{equation}
    \sum_{\mathbf{i}_\mathbf{k}\in\mathcal{I}_{3,[k_1,k_2]}}  \left| \tilde{\delta}^\nu_{\mathbf{i}_\mathbf{k}} \right|^2,
    \label{eq:Pnubar}
\end{equation}
where
\begin{equation}
    \mathcal{I}_{3,[k_1,k_2]}:=\left\{ \mathbf{i}_\mathbf{k}\in\mathcal{I}_3 \ \middle| \ k_1\le\left\|\mathbf{k}_{\mathbf{i}_\mathbf{k}}\right\|\le k_2\right\}.
\end{equation}
Thus, if we want a quantity like this, whose magnitude is larger than the Fourier-component wise amplitude, then the query complexity of our quantum algorithm might not blow up.

Eq.~(\ref{eq:Pnubar}) can be estimated by Algorithm \ref{alg:PnuEstim} with a slight modification.
Instead of Eq.~(\ref{eq:EstimProb}), we estimate
\begin{equation}
    \sum_{\mathbf{i}_\mathbf{k}\in\mathcal{I}_{3,[k_1,k_2]}} \left|\Bra{\Omega^\prime_{\mathbf{i}_\mathbf{k}}}W^\prime U_{\mathbf{f}(T),C\epsilon/4}\ket{0}^{\otimes(2\lg N_\mathrm{gr}+5)}\right|^2,
    \label{eq:EstimProbRange}
\end{equation}
that is, the probability that we obtain $\ket{0}^{\otimes(\lg N_\mathrm{gr}+5)}\ket{\mathbf{i}_\mathbf{k}}\ket{0}\ket{0}\ket{0}$ with $\mathbf{i}_\mathbf{k}\in\mathcal{I}_{3,[k_1,k_2]}$ when we measure $W^\prime U_{\mathbf{f}(T),C\epsilon/4}\ket{0}^{\otimes(2\lg N_\mathrm{gr}+5)}$.
This probability can be also estimated by QAE with query complexity in Eqs.~(\ref{eq:CompfPnuEstim1}) and (\ref{eq:CompfPnuEstim2}).

\subsubsection{Ensemble average}

So far, we have considered how to obtain $\left| \tilde{\delta}^\nu_{\mathbf{i}_\mathbf{k}} \right|^2$ for one realization, that is, the result of solving Eq.~(\ref{eq:EvEq}) with one initial value $\mathbf{f}(0)$.
However, what we really want is its ensemble average in Eq.~(\ref{eq:PnuDisc}).
In the classical Vlasov simulation, we estimate this ensemble average as
\begin{equation}
    P_\nu(\mathbf{k}_{\mathbf{i}_\mathbf{k}})\simeq \frac{1}{n_\mathrm{IV}} \sum_{i_\mathrm{IV}=0}^{n_\mathrm{IV}-1} \left| \tilde{\delta}^{\nu,i_\mathrm{IV}}_{\mathbf{i}_\mathbf{k}} \right|^2,
    \label{eq:PnuEns}
\end{equation}
that is, the sample average of $\left| \tilde{\delta}^{\nu,0}_{\mathbf{i}_\mathbf{k}} \right|^2,...,\left| \tilde{\delta}^{\nu,n_\mathrm{IV}-1}_{\mathbf{i}_\mathbf{k}} \right|^2$, the values of $\left| \tilde{\delta}^\nu_{\mathbf{i}_\mathbf{k}} \right|^2$ in the $n_\mathrm{IV}$ different runs of the simulation, where we randomly take the different initial values according to the theoretical distribution of the primordial perturbation \cite{Yoshikawa_2020,Yoshikawa2021,Yoshikawa2023}.

Also in quantum computing, we may do the same thing, by not running the quantum algorithm many times separately but utilizing quantum superposition.
That is, we extend Assumption \ref{ass:IniStOracle} so that we can generate the superposition of the initial values $\mathbf{f}^{(0)}(0),...,\mathbf{f}^{(n_\mathrm{IV}-1)}(0)$:
\begin{equation}
    O^\prime_{\mathbf{f}(0)}\ket{0}\ket{0}=\frac{1}{\sqrt{n_\mathrm{IV}}}\sum_{i_\mathrm{IV}=0}^{n_\mathrm{IV}-1}\ket{\mathbf{f}^{(i_\mathrm{IV})}(0)}\ket{i_\mathrm{IV}},
\end{equation}
where $n_\mathrm{IV}=2^{m_\mathrm{IV}}$ with $m_\mathrm{IV}\in\mathbb{N}$ for convenience.
Note that the $N$-body simulation for CDM should be also run many times with the initial values taken randomly, which leads to the different $\mathbf{F}_\mathrm{CDM}(t,\mathbf{x})$ in the different runs.
Thus, we also extend Assumption \ref{ass:FCDM} so that we can use the oracle to access $\mathbf{F}_\mathrm{CDM}(t,\mathbf{x})$ in the specified run:
\begin{align}
    \tilde{O}^{i_t}_{F_{\mathrm{CDM},x}} \ket{\mathbf{i}_\mathbf{x}}\ket{i_\mathrm{IV}}\ket{0}&= \ket{\mathbf{i}_\mathbf{x}}\ket{i_\mathrm{IV}}\ket{F_{\mathrm{CDM},x}^{i_t,\mathbf{i}_\mathbf{x},i_\mathrm{IV}}} \nonumber \\
    \tilde{O}^{i_t}_{F_{\mathrm{CDM},y}} \ket{\mathbf{i}_\mathbf{x}}\ket{i_\mathrm{IV}}\ket{0}&=  \ket{\mathbf{i}_\mathbf{x}}\ket{i_\mathrm{IV}}\ket{F_{\mathrm{CDM},y}^{i_t,\mathbf{i}_\mathbf{x},i_\mathrm{IV}}} \nonumber \\
    \tilde{O}^{i_t}_{F_{\mathrm{CDM},z}} \ket{\mathbf{i}_\mathbf{x}}\ket{i_\mathrm{IV}}\ket{0}&=  \ket{\mathbf{i}_\mathbf{x}}\ket{i_\mathrm{IV}}\ket{F_{\mathrm{CDM},z}^{i_t,\mathbf{i}_\mathbf{x},i_\mathrm{IV}}},
\end{align}
where $F_{\mathrm{CDM},x}^{i_t,\mathbf{i}_\mathbf{x},i_\mathrm{IV}}$ is $F_{\mathrm{CDM},x}(t_{i_t},\mathbf{x}_{\mathbf{i}_\mathbf{x}})$ in the $i_\mathrm{IV}$th run of the $N$-body simulation for CDM, and so on.
We postpone the implementation of these oracles to Sec.~\ref{sec:oracle}.

Equipped with these oracles, we can easily modify Algorithm \ref{alg:PnuEstim} so that it outputs Eq.~(\ref{eq:PnuEns}).
That is, we add $m_\mathrm{IV}$ qubits, and replace $O_{\mathbf{f}(0)}$ with $O^\prime_{\mathbf{f}(0)}$.
We also use $\tilde{O}^{i_t}_{F_{\mathrm{CDM},x}}$ instead of $O^{i_t}_{F_{\mathrm{CDM},x}}$, and so on.
By this change, $U_{\mathbf{f}(T),C\epsilon/4}$ in Algorithm \ref{alg:PnuEstim} is modified to the operator $U^\prime_{\mathbf{f}(T),C\epsilon/4}$ on the $(2\lg N_\mathrm{gr}+m_\mathrm{IV}+5)$-qubit system that acts as
\begin{align}
    & U^\prime_{\mathbf{f}(T),C\epsilon/4}\ket{0}^{\otimes(2\lg N_\mathrm{gr}+m_\mathrm{IV}+5)} \nonumber \\
    & \quad =\frac{1}{\sqrt{n_\mathrm{IV}}}\sum_{i_\mathrm{IV}=0}^{n_\mathrm{IV}-1}\ket{0}^{\otimes(\lg N_\mathrm{gr}+5)}\Ket{\mathbf{f}^{(i_\mathrm{IV})}(T)}\ket{i_\mathrm{IV}} + \ket{\xi_\mathrm{gar}},
\end{align}
where $\mathbf{f}^{(i_\mathrm{IV})}(T)$ is the solution of Eq.~(\ref{eq:EvEq}) for the $i_\mathrm{IV}$th initial value $\mathbf{f}^{(i_\mathrm{IV})}(0)$, and $\ket{\xi_\mathrm{gar}}$ is some unnormalized quantum state satisfying $\|\ket{\xi_\mathrm{gar}}\|\le\epsilon$.
Then, we see that the following quantity is equal to the right-hand side of Eq.~(\ref{eq:PnuEns}),
\begin{equation}
    \sum_{\mathbf{i}_\mathrm{IV}=0}^{n_\mathrm{IV}-1} \left|\Bra{\Omega^{\prime\prime}_{\mathbf{i}_\mathbf{k},i_\mathrm{IV}}}W^{\prime\prime} U^\prime_{\mathbf{f}(T),C\epsilon/4}\ket{0}^{\otimes(2\lg N_\mathrm{gr}+m_\mathrm{IV}+5)}\right|^2,
    \label{eq:ProbEstEns}
\end{equation}
where $\Ket{\Omega^{\prime\prime}_{\mathbf{i}_\mathbf{k},i_\mathrm{IV}}}:=\Ket{\Omega^{\prime}_{\mathbf{i}_\mathbf{k}}}\ket{i_\mathrm{IV}}$ and $W^{\prime\prime}:=W^\prime \otimes I_{n_\mathrm{IV}}$.
Eq.~(\ref{eq:ProbEstEns}) is the probability that we obtain $\ket{0}^{\otimes(\lg N_\mathrm{gr}+5)}\ket{\mathbf{i}_\mathbf{k}}\ket{0}\ket{0}\ket{0}\ket{i_\mathrm{IV}}$ with any $i_\mathrm{IV}\in[n_\mathrm{IV}]_0$ when we measure $W^{\prime\prime} U^\prime_{\mathbf{f}(T),C\epsilon/4}\ket{0}^{\otimes(2\lg N_\mathrm{gr}+m_\mathrm{IV}+5)}$ and can be estimated by QAE.
The numbers of queries to the replaced oracles are still of order (\ref{eq:CompfPnuEstim1}) and (\ref{eq:CompfPnuEstim2}).

Note that in the above discussion, we have implicitly assumed that the value of $C$ is the same for different $\mathbf{f}^{(i_\mathrm{IV})}(0)$.
This in fact holds approximately for large $N_\mathrm{gr}$, since $C$ is written with the averages of the values of $f(0,\mathbf{x},\mathbf{v})$ and its square on the $N_\mathrm{gr}$ grid points as Eq.~(\ref{eq:C2}), and the realizations of $\mathbf{f}^{(i_\mathrm{IV})}(0)$ for various $i_\mathrm{IV}$ are generated based on the same distribution of the primordial perturbation.

\subsection{Implementation of the oracles \label{sec:oracle}}

Now, let us consider the implementations of the oracles used in our quantum algorithm.

\subsubsection{Oracles to access $\mathbf{F}_\mathrm{CDM}$}

We first consider $\left\{O^{i_t}_{F_{\mathrm{CDM},x}}, O^{i_t}_{F_{\mathrm{CDM},y}}, O^{i_t}_{F_{\mathrm{CDM},z}}\right\}_{i_t}$, the oracles to access $\mathbf{F}_\mathrm{CDM}(t,\mathbf{x})$ the gravitational force by CDM as Eq.~(\ref{eq:OracleFCDM}).
Since we now assume that the $N$-body simulation for CDM gives the values of this on the $n_\mathrm{gr}^3$ grid points in the position coordinates, it seems that we need to resort to QRAM that stores those $n_\mathrm{gr}^3$ real numbers.
In the recent Vlasov simulations with massive neutrino, the number of the grid points in the position coordinates is of order $10^9$ \cite{Yoshikawa2021}.
Because of the issues on QRAM mentioned in Sec.~\ref{sec:QRAM}, realizing a QRAM with such a large number of entries seems challenging.
Nevertheless, compared to the classical Vlasov simulation, in which the $O(n_\mathrm{gr}^6)$ memory space is used, the QRAM size of $O(n_\mathrm{gr}^3)$ indicates a large improvement with respect to scaling on $n_\mathrm{gr}$.

We should also note that calculating $\{F_{\mathrm{CDM},x}^{i_t,\mathbf{i}_\mathbf{x}},F_{\mathrm{CDM},y}^{i_t,\mathbf{i}_\mathbf{x}},F_{\mathrm{CDM},z}^{i_t,\mathbf{i}_\mathbf{x}}\}_{i_t,\mathbf{i}_\mathbf{x}}$ by the $N$-body simulation and preparing QRAMs that store them takes $O(n_tn_\mathrm{gr}^3)$ time, which has the worse scaling on $n_\mathrm{gr}$ than the query complexity of our quantum algorithm in Eq.~(\ref{eq:CompfPnuEstim1}).
However, as mentioned in Sec.~\ref{sec:ClLSSSim}, in the current classical computing, the $N$-body simulation for only CDM is less heavy than the Vlasov and $N$-body simulations including massive neutrino.
Therefore, it is reasonable to consider only the application of quantum computing to the Vlasov simulation for neutrino, leaving the $N$-body simulation for CDM classical.
Also note that, once QRAMs storing $F_{\mathrm{CDM},x}^{i_t,\mathbf{i}_\mathbf{x}}$ and so on are prepared, we can reuse it for the neutrino Vlasov simulations in the various settings, e.g., various neutrino masses, which relatively diminishes the cost for the $N$-body simulation for CDM.

We also comment that, the approximation that $\mathbf{F}_\mathrm{CDM}(t,\mathbf{x})$ is piecewise constant in time reduces the size of QRAMs to be prepared.
Although we are given $\{F_{\mathrm{CDM},x}^{i_t,\mathbf{i}_\mathbf{x}},F_{\mathrm{CDM},y}^{i_t,\mathbf{i}_\mathbf{x}},F_{\mathrm{CDM},z}^{i_t,\mathbf{i}_\mathbf{x}}\}_{i_t,\mathbf{i}_\mathbf{x}}$, whose number per vector element, say $F_{\mathrm{CDM},x}$, is $n_tn_\mathrm{gr}^3$, our algorithm uses not one QRAM with $n_tn_\mathrm{gr}^3$ entries per vector element, but $n_t$ QRAMs with $n_\mathrm{gr}^3$ entries, which are $\{F_{\mathrm{CDM},x}^{i_t,\mathbf{i}_\mathbf{x}}\}_{\mathbf{i}_\mathbf{x}}$ for one value of $i_t\in[n_t]_0$.
This is because, as explained in Sec.~\ref{sec:GenQSHamSim}, to get $\ket{\mathbf{f}(T)}$, we successively apply Hamiltonian simulations for the time intervals $[t_0,t_1),...,[t_{n_t-1},t_{n_t})$, and in each of them we use only the constant value of $\mathbf{F}_\mathrm{CDM}(t,\mathbf{x})$ in the interval.
Dividing the QRAM into smaller ones reduces the difficulty of implementation.
In particular, for the circuit-based QRAM, we do not need to use $n_tn_\mathrm{gr}^3$ qubits, but it suffices to reuse $n_\mathrm{gr}^3$ qubits, with different QRAMs implemented in the Hamiltonian simulation for the different time interval.
We do not enjoy this reduction if $\mathbf{F}_\mathrm{CDM}(t,\mathbf{x})$ continuously varies in time and we use the time-dependent Hamiltonian simulation \cite{Kieferova2019}, in which we use the oracle $O^H_\mathrm{ent}:\ket{t}\ket{i}\ket{j}\ket{0}\mapsto\ket{t}\ket{i}\ket{j}\ket{H_{ij}(t)}$ and thus the oracles such as $O_{F_{\mathrm{CDM},x}}: \ket{t}\ket{\mathbf{i}_\mathbf{x}}\ket{0} \mapsto \ket{t}\ket{\mathbf{i}_\mathbf{x}}\ket{F_{\mathrm{CDM},x}(t,x_{\mathbf{i}_\mathbf{x}})}$ in our case.

As the last comment, we note that when we want to estimate Eq.~(\ref{eq:PnuEns}), the average of $\left| \tilde{\delta}^{\nu,0}_{\mathbf{i}_\mathbf{k}} \right|^2,...,\left| \tilde{\delta}^{\nu,n_\mathrm{IV}-1}_{\mathbf{i}_\mathbf{k}} \right|^2$, the number of the entries in each QRAM increases to $n_\mathrm{gr}^3n_\mathrm{IV}$.

\subsubsection{Oracle to generate the state encoding the initial value}

Next, we consider the oracle $O_{\mathbf{f}(0)}$ in Eq.~(\ref{eq:Of0}), which generates the quantum state $\ket{\mathbf{f}(0)}$ encoding the initial value $\mathbf{f}(0)$.

In the neutrino Vlasov simulation, the initial distribution function is set to a Fermi-Dirac distribution \cite{Yoshikawa2023},
\begin{equation}
f(0,\mathbf{x},\mathbf{v}) \propto (1+\delta_\nu(0,\mathbf{x}))F_\mathrm{FD}(\mathbf{v}-\mathbf{v}_\mathrm{b}(\mathbf{x})).
\label{eq:NeuIni}
\end{equation}
Here, $\mathbf{v}_\mathrm{b}(\mathbf{x})$ is the neutrino bulk velocity at position $\mathbf{x}$, and
\begin{equation}
    F_\mathrm{FD}(\mathbf{v}-\mathbf{v}_\mathrm{b}):=\frac{1}{\exp\left(\frac{m_\nu|\mathbf{v}-\mathbf{v}_\mathrm{b}|}{k_\mathrm{B}T_\nu}\right)+1},
    \label{eq:F_FD}
\end{equation}
where $k_\mathrm{B}$ is the Boltzmann constant and $T_\nu$ is the neutrino temperature at the initial time.
The initial density fluctuation and bulk velocity are calculated based on the theory for the primordial cosmological perturbation \cite{Yoshikawa2023}.

Then, our aim is generating the following quantum state $\ket{\mathbf{f}(0)}$ encoding Eq.~(\ref{eq:NeuIni}).
It is written as
\begin{equation}
    \ket{\mathbf{f}(0)} = \mathcal{C} \sum_{\mathbf{i}_\mathbf{x}\in\mathcal{I}_3} (1+\delta_\nu(0,\mathbf{x}_{\mathbf{i}_\mathbf{x}}))\ket{\mathbf{i}_\mathbf{x}}\ket{F_{\mathrm{FD},\mathbf{v}_\mathrm{b}(\mathbf{x}_{\mathbf{i}_\mathbf{x}})}},
\end{equation}
where
\begin{equation}
    \ket{F_{\mathrm{FD},\mathbf{v}_\mathrm{b}}}:= \mathcal{D}_{\mathbf{v}_\mathrm{b}}\sum_{\mathbf{i}_\mathbf{v}\in\mathcal{I}_3} F_\mathrm{FD}(\mathbf{v}_{\mathbf{i}_\mathbf{v}}-\mathbf{v}_\mathrm{b}) \ket{\mathbf{v}_{\mathbf{i}_\mathbf{v}}}.
\end{equation}
and $\mathcal{C}$ and $\mathcal{D}_{\mathbf{v}_\mathrm{b}}$ are normalization constants.

We now assume that $\{\delta_\nu(0,\mathbf{x}_{\mathbf{i}_\mathbf{x}})\}_{\mathbf{i}_\mathbf{x}}$ are classically computed as a preparation of our quantum algorithm and stored in a QRAM.
Given such a QRAM, we can generate the quantum state $\mathcal{C} \sum_{\mathbf{i}_\mathbf{x}\in\mathcal{I}_3} (1+\delta_\nu(0,\mathbf{x}_{\mathbf{i}_\mathbf{x}}))\ket{\mathbf{i}_\mathbf{x}}$ querying the QRAM $O(\log n_\mathrm{gr}^3)$ times by the method in \cite{kerenidis2017quantum}.
We also assume that $\{\mathbf{v}_\mathrm{b}(\mathbf{x}_{\mathbf{i}_\mathbf{x}})\}_{\mathbf{i}_\mathbf{x}}$ are precomputed and stored in a QRAM.
There are some methods to generate a quantum state encoding a function given as an explicit formula in the amplitudes \cite{grover2002creating,Sanders2019,marin2021quantum,rattew2022preparing}.
By such a method, we can perform the operation $\ket{\mathbf{v}_\mathrm{b}}\ket{0}\rightarrow\ket{\mathbf{v}_\mathrm{b}}\ket{F_{\mathrm{FD},\mathbf{v}_\mathrm{b}}}$, querying the quantum circuit to compute $F_\mathrm{FD}(\mathbf{v}-\mathbf{v}_\mathrm{b})$ $O(\log n_\mathrm{gr}^3)$ times.
Combining these circuits and QRAMs, we can generate $\ket{\mathbf{f}(0)}$ with $O(\mathrm{polylog}(n_\mathrm{gr}))$ query complexity.
Note that the sizes of the aforementioned QRAMs are $O(n_\mathrm{gr}^3)$, which is of the same order as QRAMs storing $\mathbf{F}_\mathrm{CDM}$.
Also note that preparing those QRAMs takes a $O(n_\mathrm{gr}^3)$ time and this does not increase the order of the preparation time, which already includes the $O(n_\mathrm{gr}^3)$ time to prepare the QRAMs for $\mathbf{F}_\mathrm{CDM}$.

In the above discussion, we have assumed that there is only one neutrino flavor with mass $m_\nu$.
This is just for simplicity, and the extension to the more realistic multi-flavor setting is straightforward as follows.
The point is that we are now approximating the neutrino Vlasov equation as Eq.~(\ref{eq:VlasovNeu}) and thus the different masses of the different flavors affect the solution only through the initial value, which depends on the mass.
Since Eq.~(\ref{eq:VlasovNeu}) is linear with respect to $f$, if we set the initial value to the linear combination of those for the various flavors, then the solution becomes the distribution function for the mixture of the flavors.
In our quantum algorithm, the modification is only replacing $\ket{F_{\mathrm{FD},\mathbf{v}_\mathrm{b}}}$, which encodes the single-flavor distribution, with the quantum state encoding the initial distribution of the mixture.
Since it is still given as an explicit formula, we can use the aforementioned methods for function-loading to the quantum state.

Lastly, we make a comment on the cost to calculate $C$, which has been postponed.
The numerator of Eq.~(\ref{eq:C2}) is equal to
\begin{equation}
    \frac{1}{64V^6}\left(\bar{\rho}_\nu(0)\right)^2,
\end{equation}
and thus written with $\bar{\rho}_\nu(0)$, which is determined by the normalization of $f$.
On the other hand, the denominator of Eq.~(\ref{eq:C2}) becomes
\begin{equation}
    \frac{1}{N_\mathrm{gr}}\sum_{\mathbf{i}_\mathbf{x}\in\mathcal{I}_3}(1+\delta_\nu(0,\mathbf{x}_{\mathbf{i}_\mathbf{x}}))^2\sum_{\mathbf{i}_\mathbf{v}\in\mathcal{I}_3} \left(F_\mathrm{FD}(\mathbf{v}_{\mathbf{i}_\mathbf{v}}-\mathbf{v}_\mathrm{b}(\mathbf{x}_{\mathbf{i}_\mathbf{x}}))\right)^2
    \label{eq:CDenom}
\end{equation}
for $f(0,\mathbf{x},\mathbf{v})$ in the form of Eq.~(\ref{eq:NeuIni}).
The latter part of Eq.~(\ref{eq:CDenom}) is calculated as
\begin{align}
    & \sum_{\mathbf{i}_\mathbf{v}\in\mathcal{I}_3} \left(F_\mathrm{FD}(\mathbf{v}_{\mathbf{i}_\mathbf{v}}-\mathbf{v}_\mathrm{b}(\mathbf{x}_{\mathbf{i}_\mathbf{x}}))\right)^2 \nonumber \\
    \simeq & \frac{1}{\Delta_\mathbf{v}^3}\int_{-V}^V du \int_{-V}^V dv \int_{-V}^V dw \left(F_\mathrm{FD}(\mathbf{v}-\mathbf{v}_\mathrm{b}(\mathbf{x}_{\mathbf{i}_\mathbf{x}}))\right)^2 \nonumber \\
    \simeq & \frac{1}{\Delta_\mathbf{v}^3}\int_{-\infty}^\infty du \int_{-\infty}^\infty dv \int_{-\infty}^\infty dw \left(F_\mathrm{FD}(\mathbf{v})\right)^2,
    \label{eq:FFDSq}
\end{align}
where we replace the sum with the integral at the first approximation, and at the second approximation we neglect the contribution to the integral from the region outside $[-V,V]\times[-V,V]\times[-V,V]$.
Eq.~(\ref{eq:FFDSq}) is independent of $\mathbf{i}_\mathbf{x}$ and evaluated numerically.
The remaining part is $\sum_{\mathbf{i}_\mathbf{x}\in\mathcal{I}_3}(1+\delta_\nu(0,\mathbf{x}_{\mathbf{i}_\mathbf{x}}))^2$.
We can calculate this at the same time as generating $\{\delta_\nu(0,\mathbf{x}_{\mathbf{i}_\mathbf{x}})\}_{\mathbf{i}_\mathbf{x}}$, keeping the order of the preparation time $O(n_\mathrm{gr}^3)$.

\section{Demonstration \label{sec:demo}}

\begin{figure*}[tp]
\centering
    \subfigure[$T=0$]{
    \includegraphics[width=1\columnwidth]{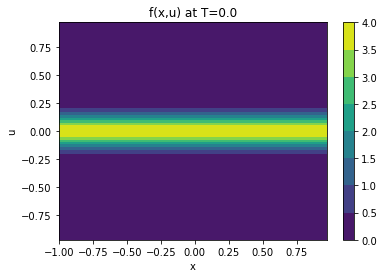} 	\label{fig:fT0}
    }
    \subfigure[$T=0.1$]{
    \includegraphics[width=1\columnwidth]{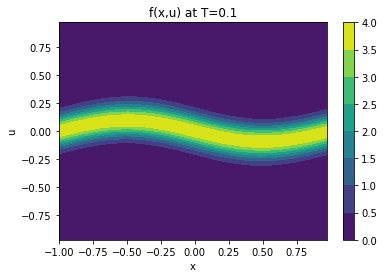} 	\label{fig:fT01}
    }
    \subfigure[$T=0.2$]{
    \includegraphics[width=1\columnwidth]{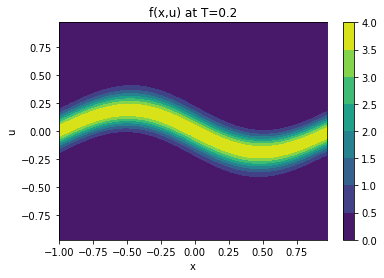} 	\label{fig:fT02}
    }

\caption{$f(T,x,u)$ at $T=0,0.1$ and $0.2$.}
\label{fig:fHeatmap}
\end{figure*}

\begin{figure}[tp]
\centering
\includegraphics[width=1\columnwidth]{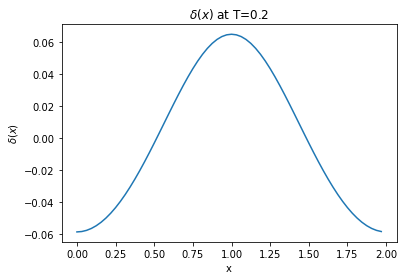}
\caption{$\delta^\nu_{i_x}$ at $T=0.2$.}
\label{fig:deltax}
\end{figure}

\begin{figure}[tp]
\centering
\includegraphics[width=1\columnwidth]{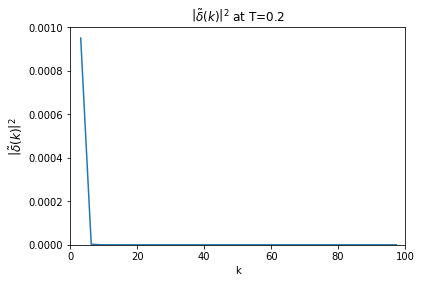}
\caption{$\left|\tilde{\delta}^\nu_{i_k}\right|^2$ at $T=0.2$.}
\label{fig:deltak}
\end{figure}

Lastly, we present an illustrative numerical model.
It is impossible to run our algorithm on a current real quantum computer or a quantum circuit simulator.
Instead, we demonstrate a core part of our algorithm, namely, the Hamiltonian simulation-based time evolution of the neutrino distribution function $f(t,\mathbf{x},\mathbf{v})$.
By considering a toy problem, we find the solution of Eq.~(\ref{eq:ShEq}), the discretized Vlasov equation, as Eq.~(\ref{eq:fTExp}) with the matrix exponentiation $\exp\left(-i \Delta_t H_{i_t}\right)$ calculated explicitly on a classical computer.
This is what we aim at by the quantum algorithm of Hamiltonian simulation, and differs from common classical methods for time integration.

We consider the 1D setting with the position coordinate $x$ and the velocity coordinate $v$.
We give the CDM gravity by the following explicit function
\begin{equation}
    F_{\mathrm{CDM},x}(x)=A\sin(Kx)
\end{equation}
with real constants $A$ and $K$.
Then, under the discretization described in Sec. \ref{sec:Disc}, we find $\mathbf{f}(T)=\left(f(T,x_0,u_0),...,f(T,x_{n_\mathrm{gr}-1},u_{n_\mathrm{gr}-1})\right)$, the vector of the values of $f(T,x,u)$ on the grid points.
Since $F_{\mathrm{CDM},x}$ does not depend on time and thus the Hamiltonian $H$ does not either, we can get $\mathbf{f}(T)$ by applying a single operator $\exp\left(-i HT\right)$ to $\mathbf{f}(0)$:
\begin{equation}
    \ket{\mathbf{f}(T)}=\exp\left(-i HT\right)\ket{\mathbf{f}(0)}.
    \label{eq:fTExpOneTime}
\end{equation}
The initial value is set as follows: $f(0,x,u)$ is constant in $x$, and the Maxwell distribution is adopted in the $u$ direction.
That is,
\begin{equation}
    f(0,x,u)=\frac{1}{\sqrt{2\pi\sigma_v^2}} \exp\left(-\frac{u^2}{2\sigma_v^2}\right)
\end{equation}
with $\sigma_v>0$.

We show the result for $A=-1$, $K=\pi$ and $\sigma_v=0.1$ under the discretization with $L=2$, $V=1$ and $n_\mathrm{gr}=64$.
FIG. \ref{fig:fHeatmap} shows the heatmaps of $f(T,x,u)$ drawn with $\mathbf{f}(T)$ obtained by Eq.~(\ref{eq:fTExpOneTime}) with $T=0,0.1$, and $0.2$.
From this $\mathbf{f}(T)$, we calculate $\delta^\nu_{i_x}$, the neutrino density perturbation at each grid point, as Eq.~(\ref{eq:delNu}), and plot it as a function of $x=i_xL/n_\mathrm{gr}$ in FIG.~\ref{fig:deltax}.
Because of the functional form of $F_{\mathrm{CDM},x}$, we expect that the density is enhanced at $x=1$ and suppressed at $x=0,2$, and FIG.~\ref{fig:deltax} matches this expectation.
Using this $\delta^\nu_{i_x}$, we calculate $|\tilde{\delta}^\nu_{i_k}|^2$, the squared amplitude of the Fourier mode of the neutrino density perturbation, via Eq.~(\ref{eq:tilDelNu}), and plot it as a function of $k=2 \pi i_k/L$ in FIG.~\ref{fig:deltak}.
Since the neutrino distribution evolves under the force as a single sinusoidal function with wavenumber $k=K$, we expect that only the Fourier mode with that wavenumber is enhanced.
In fact, this is observed in FIG.~\ref{fig:deltak}.
In summary, these figures imply that our calculation is working in this demonstration.

\section{Summary \label{sec:summary}}

In this paper, we considered the quantum algorithm for simulations of LSS formation with massive neutrinos. The large-scale neutrino distribution is an important issue for both cosmology and particle physics.
In order to follow the growth of density perturbations of neutrinos with a large velocity dispersion, 
it is desirable to solve directly
 the Vlasov equation in an efficient way 
 rather than performing conventional $N$-body simulations.
However, one needs to solve the PDE in a $(6+1)$-dimensional space, and thus it is a challenging task: Taking $n_\mathrm{gr}$ grid points in each of six space coordinates and $n_t$ time grid points leads to $O(n_tn_\mathrm{gr}^6)$ and $O(n_\mathrm{gr}^6)$ space complexity.
We thus proposed a quantum algorithm for this task.
First, by neglecting the gravity generated by the neutrinos
that contribute only a tiny fraction in mass, and hence by taking into account only the gravity by CDM, we approximated the Vlasov equation in a linear form.
We showed that the discretized Vlasov equation can be regarded as a Schr\"{o}dinger equation.
Then, by calculating the distribution of CDM using a fast $N$-body simulation in advance, we applied the Hamiltonian simulation to solve the equation.
We proposed not only how to obtain the solution-encoding quantum state but also a way to extract the neutrino power spectrum, which is an important quantity for comparisons with cosmological observations and thus is of practical interest, from the quantum state by QAE.
Our method outputs a $\epsilon$-approximation of the power spectrum with $\widetilde{O}((n_\mathrm{gr}+n_t)/\epsilon)$ query complexity, using $O\left(\log^{5/2}\left(n_\mathrm{gr}/\epsilon\right)\right)$ qubits.

Although our algorithm has some shortcomings such as the assumption on the availability of large-sized QRAMs, it is the first proposal of a quantum algorithm for the LSS simulation that outputs the quantity of practical interest with guaranteed accuracy.
We thus believe that our algorithm is an important step in the application of quantum computing to the LSS simulation or, more broadly, challenging numerical tasks in cosmology and particle physics.

In future works, we will explore the possibility of extending our algorithm.
Obviously, taking into account the neutrino self-gravity via, e.g., Carleman linearization, is one possible way.
Another possibility is incorporating higher-order schemes for partial derivatives, for which we are now adopting the central difference as Eq.~(\ref{eq:centDiff}).
Adopting higher-order schemes has two effects.
On the one hand, it makes the coefficient matrix $A(t)$ less sparse: If we adopt a $n$-th-order scheme, then the sparsity $s$ scales as $O(n)$.
In the sparse-access setting, the query complexity of Hamiltonian simulation scales as $O(s)$ on $s$, and thus so does our algorithm.
In total, the query complexity of our algorithm scales as $O(n)$ on $n$.
On the other hand, by adopting a higher-order scheme, we may configure spatial grid points with larger interval, which means smaller $n_\mathrm{gr}$, keeping the accuracy of the solution.
Since the query complexity of our algorithm scales as $O(n_\mathrm{gr})$, in total, if an $n$-th-order scheme leads to $n_\mathrm{gr}$ scaling better than $n_\mathrm{gr}=O(1/n)$, then higher-order schemes would be beneficial.
As a previous study on the effect of higher-order schemes in practice, we refer to \cite{Tanaka_2017}, which performed a set of simulations with third-, fifth-, and seventh-order schemes and in fact observed the improvement of the accuracy.
However, in general, the extent to which higher-order schemes improve accuracy is highly problem-dependent, and thus we leave incorporating higher-order schemes into our algorithm for future works.

\section*{acknowledgement}

KM is supported by MEXT Quantum Leap Flagship Program (MEXT Q-LEAP) Grant no. JPMXS0120319794, JSPS KAKENHI Grant no. JP22K11924, and JST COI-NEXT Program Grant No. JPMJPF2014.

SY is supported by the Forefront Physics and Mathematics Program to Drive Transformation (FoPM).

FU is supported by JSPS KAKENHI Grant No.~JP23KJ0642 and the Forefront Physics and Mathematics Program to Drive Transformation (FoPM).

KF is supported by JSPS KAKENHI Grant No.~JP20K14512.

\appendix

\section{Proofs}

\subsection{Building-block theorems \label{sec:theorem}}

Here, we present some theorems from previous papers, which are used to prove our main results, theorems on the complexity of our quantum algorithm.

\subsubsection{Block encoding}

Mathematically, block-encoding is defined as follows:

\begin{definition}[\cite{Gilyen2019}, Definition 24]
Let $s,a\in\mathbb{N}$, $\alpha,\epsilon\in\mathbb{R}_+$ and $A\in\mathbb{C}^{2^s\times2^s}$.
We say that a unitary $U$ on a $(s+a)$-qubit system is an $(\alpha,a,\epsilon)$-block-encoding of $A$, if
\begin{equation}
    \left\|A-\alpha(\bra{0}^{\otimes a}\otimes I_{2^s})U(\ket{0}^{\otimes a}\otimes I_{2^s})\right\|\le\epsilon
\end{equation}
\end{definition}

We can efficiently construct a block-encoding of a sparse matrix if we have access to its entries.
The cost for this is summarized as the following theorem.

\begin{theorem}[\cite{Gilyen2019}, Lemma 48 in the full version, modified]

Let $A=(A_{ij})\in\mathbb{C}^{2^w \times 2^w}$ be a $s$-sparse matrix.
Suppose that we have accesses to the oracles $O^A_\mathrm{row}$ and $O^A_\mathrm{col}$ that act on a two register system as
\begin{equation}
    O^A_\mathrm{row}\ket{i}\ket{k}=\ket{i}\ket{r_{ik}},O^A_\mathrm{col}\ket{k}\ket{i}=\ket{c_{ki}}\ket{i}
\end{equation}
for any $i\in[2^w]_0$ and $k\in[s]_0$, where $r_{ik}$ (resp. $c_{ki}$) is the index of the $k$th nonzero entry in the $i$th row (resp. $i$th column) in $A$, or $i+2^w$ if there are less than $k$ nonzero entries.
Besides, suppose that we have accesses to the oracle $O^A_\mathrm{ent}$ that acts on a three register system as
\begin{equation}
    O^A_\mathrm{ent}\ket{i}\ket{j}\ket{0}=\ket{i}\ket{j}\ket{A_{ij}}.
\end{equation}
Then, for any $\epsilon\in\mathbb{R}_+$, there exists a $(s\|A\|_\mathrm{max},w+3,\epsilon)$-block-encoding of $A$, in which $O^A_\mathrm{row}$ and $O^A_\mathrm{col}$ are each queried once, $O^A_\mathrm{ent}$ is queried twice, additional $O\left(w+\log^{5/2}\left(\frac{s^2\|A\|_\mathrm{max}}{\epsilon}\right)\right)$ one and two qubit gates are used, and $O\left(\log^{5/2}\left(\frac{s^2\|A\|_\mathrm{max}}{\epsilon}\right)\right)$ ancilla qubits are used.
\label{th:BlEncSp}
\end{theorem}

\subsubsection{Hamiltonian simulation}

The theorem on the complexity of Hamiltonian simulation, which was postponed in Sec. \ref{sec:qalgo}, is as follows.

\begin{theorem}[\cite{Gilyen2019}, Corollary 62 in the full version]
    Let $\epsilon\in\left(0,\frac{1}{2}\right)$, $t\in\mathbb{R}$, $\alpha\in\mathbb{R}_+$ and $s,a\in\mathbb{N}$.
    Let $H$ be a $2^s \times 2^s$ Hermitian matrix and $U$ be a $(\alpha,a,\epsilon/|2t|)$-block encoding of $H$.
    Then, we can implement $(1,a+2,\epsilon)$-block-encoding of $\exp(-itH)$ using $U$ or its inverse  $O(6\alpha|t|+9\log(12/\epsilon))$ times, controlled $U$ or its inverse 3 times, additional $O(a(\alpha|t|+\log(2/\epsilon)))$ two-qubit gates, and $O(1)$ ancilla qubits.
    \label{th:BlEncExpH}
\end{theorem}

\subsubsection{Quantum amplitude estimation}

The complexity of QAE is evaluated as the following theorem.

\begin{theorem}[\cite{brassard2002}, Theorem 12, modified]
    Suppose that we are given an access to a quantum circuit $A$ that acts on a quantum register as $A\ket{0}=\ket{\Phi}:=\sqrt{a}\ket{\phi}+\sqrt{1-a}\ket{\phi_\perp}$, where $\ket{\phi}$ and $\ket{\phi_\perp}$ are mutually orthogonal states and $a\in(0,1)$.
    Also, suppose that we have access to a quantum circuit $V$ on the same system that acts as $V\ket{\phi}=-\ket{\phi}$ and $V\ket{\phi_\perp}=\ket{\phi_\perp}$. 
    Then, for any $\epsilon\in\mathbb{R}_+$ and $\delta\in(0,1)$, there exists a quantum algorithm that with probability at least $1-\delta$ outputs a $\epsilon$-approximation of $a$ calling $A$ and $V$
    \begin{equation}
        O\left(\frac{1}{\epsilon}\log\left(\frac{1}{\delta}\right)\right)
    \end{equation}
    times.
    \label{th:QAE}
\end{theorem}

In this paper, we use QAE for the case that the target state $\ket{\phi}$ is a superposition of some set $\mathcal{S}$ of computational basis states and $\ket{\phi_\perp}$ is a superposition of the other ones.
In this case, $a$ is the probability that we get a bit string corresponding to any state in $\mathcal{S}$ when we measure $\ket{\Phi}$.

Note that, although the success probability in the original theorem in \cite{brassard2002} is a constant, it is $1-\delta$ in Theorem \ref{th:QAE} with a factor $\log(1/\delta)$ added in the query complexity bound.
This is due to taking the median of the results in multiple runs of the algorithm \cite{montanaro2015}, which is based on the powering lemma in \cite{jerrum1986}.

We also comment on the qubit number.
The original version of QAE proposed in \cite{brassard2002} is based on quantum phase estimation (QPE) \cite{kitaev1995quantum}, and uses $O(\log (1/\epsilon))$ qubits to output the estimate of $a$ in addition to qubits used in $A$ and $V$.
Furthermore, it requires the controlled versions of $A$ and $V$, for which we may need to use additional qubits and gates compared to the uncontrolled ones.
Fortunately, after \cite{brassard2002}, many variants without QPE, which require neither additional qubits nor controlled oracles, have been proposed \cite{suzuki2020amplitude,grinko2021iterative,Aaronson2020,Fukuzawa2023}.
Thus, we hereafter consider that in QAE we use only qubits needed to operate $A$ and $V$.

\subsection{Proof of Theorem \ref{th:ketfTGen} \label{sec:PrThKetfT}}

As a preparation, let us see that we can implement the oracles used to construct the block-encoding of $H$.

\begin{lemma}
    Suppose that Assumption \ref{ass:FCDM} holds.
    Then, for $H(t)$ in Eq.~(\ref{eq:H}) with any $t\in[0,T)$, we have $O^{H(t)}_\mathrm{row}$, $O^{H(t)}_\mathrm{col}$ and $O^{H(t)}_\mathrm{ent}$, in which $\left\{O^{i_t}_{F_{\mathrm{CDM},x}}, O^{i_t}_{F_{\mathrm{CDM},y}}, O^{i_t}_{F_{\mathrm{CDM},z}}\right\}_{i_t}$ and arithmetic circuits are queried $O\left(1\right)$ times.
    \label{lem:OH}
\end{lemma}

\begin{proof}
     If we have $O^{A(t)}_\mathrm{row}$, $O^{A(t)}_\mathrm{col}$, and $O^{A(t)}_\mathrm{ent}$ for $A(t)$ in Eq.~(\ref{eq:A}), then we immediately have those for $H(t)$ in Eq.~(\ref{eq:H}), since the latter is just the former multiplied by $i$.
    Thus, we hereafter consider $O^{A(t)}_\mathrm{row}$, $O^{A(t)}_\mathrm{col}$, and $O^{A(t)}_\mathrm{ent}$.

    Let us start from $O^{A(t)}_\mathrm{row}$.
    For $A(t)$, whose sparsity is 12, $r_{ik}$ is given as\footnote{Strictly speaking, this expression for $r_{ik}$ holds for only $i$ such that none of the entries of $\sigma^{-1}(i)$ is $0$ or $n_\mathrm{gr}-1$, and otherwise the expression is slightly modified. However, such a handling is straightforward and thus we do not show the complete expression here for conciseness.}
    \begin{equation}
        r_{ik}=\sigma(\sigma^{-1}(i)+\mathbf{d}_k)
    \end{equation}
    with
    \begin{align}
        \mathbf{d}_1&=(-1,0,0,0,0,0) \nonumber \\
        \mathbf{d}_2&=(1,0,0,0,0,0) \nonumber \\
        & \vdots \nonumber \\
        \mathbf{d}_{11}&=(0,0,0,0,0,-1) \nonumber \\
        \mathbf{d}_{12}&=(0,0,0,0,0,1).
    \end{align}
    Here, $\sigma:\mathcal{I}_6\rightarrow[N_\mathrm{gr}]_0$ is the map in Eq.~(\ref{eq:iandVeci}) and $\sigma^{-1}$ is its inverse.
    Note that, in the binary representation on qubits, these are in fact doing nothing, or just changing whether we consider that a computational basis state corresponds to $\mathbf{i}\in\mathcal{I}_6$ or $\sigma(\mathbf{i})\in[N_\mathrm{gr}]_0$.
    Thus, $O^{A(t)}_\mathrm{row}$ is implemented with controlled adders, thus as a combination of arithmetic circuits.
    $O^{A(t)}_\mathrm{col}$ is implemented similarly.

    On the implementation of $O^{A(t)}_\mathrm{ent}$, we consider each component of $A(t)$ separately.
    Let us first consider $A_x$.
    For $i,j\in[N_\mathrm{gr}]_0$, we can perform the following operation:
    \begin{align}
        &\ket{i}\ket{j}\ket{0}\ket{0}\ket{0} \nonumber \\
        \rightarrow & \ket{i}\ket{j}\ket{\mathbf{1}_{\sigma^{-1}(i)-\sigma^{-1}(j)=(1,0,0,0,0,0)}}  \nonumber \\
        &\otimes \ket{\mathbf{1}_{\sigma^{-1}(i)-\sigma^{-1}(j)=(-1,0,0,0,0,0)}}\ket{0} \nonumber \\
        = & \ket{i}\ket{j}\left(\mathbf{1}_{\sigma^{-1}(i)-\sigma^{-1}(j)\ne(\pm1,0,0,0,0,0)}\ket{0}\ket{0}\right.  \nonumber \\
        & \qquad \quad +\mathbf{1}_{\sigma^{-1}(i)-\sigma^{-1}(j)=(1,0,0,0,0,0)}\ket{1}\ket{0} \nonumber \\
        & \qquad \quad \left.+\mathbf{1}_{\sigma^{-1}(i)-\sigma^{-1}(j)=(-1,0,0,0,0,0)}\ket{0}\ket{1}\right)\ket{0} \nonumber \\
        \rightarrow & \ket{i}\ket{j} \nonumber \\
        & \otimes\left(\mathbf{1}_{\sigma^{-1}(i)-\sigma^{-1}(j)\ne(\pm1,0,0,0,0,0)}\ket{0}\ket{0}\ket{0}\right.  \nonumber \\
        & +\mathbf{1}_{\sigma^{-1}(i)-\sigma^{-1}(j)=(1,0,0,0,0,0)}\ket{1}\ket{0}\Ket{\frac{V-(i_u+1)\Delta_\mathbf{v}}{2\Delta_\mathbf{x}}} \nonumber \\
        & \left.+\mathbf{1}_{\sigma^{-1}(i)-\sigma^{-1}(j)=(-1,0,0,0,0,0)}\ket{0}\ket{1}\Ket{\frac{-V+(i_u+1)\Delta_\mathbf{v}}{2\Delta_\mathbf{x}}}\right) \nonumber \\
        = & \ket{i}\ket{j} \nonumber \\
        & \otimes\ket{\mathbf{1}_{\mathbf{i}(i)-\mathbf{i}(j)=(1,0,0,0,0,0)}}\ket{\mathbf{1}_{\mathbf{i}(i)-\mathbf{i}(j)=(-1,0,0,0,0,0)}}\ket{(A_x)_{ij}}  \nonumber \\
        \rightarrow & \ket{i}\ket{j}\ket{0}\ket{0}\ket{(A_x)_{ij}}.
        \label{eq:OAx}
    \end{align}
    Here, the transformation at the first arrow is done by some arithmetic circuits.
    At the second arrow, regarding $\ket{i}$ as $\ket{\sigma^{-1}(i)}=\ket{i_x}\ket{i_y}\ket{i_z}\ket{i_u}\ket{i_v}\ket{i_w}$, we use the controlled versions of
    \begin{align}
        U_{A_x,1}:\ket{i_u}\ket{0}&\mapsto\ket{i_u}\Ket{\frac{V-(i_u+1)\Delta_\mathbf{v}}{2\Delta_\mathbf{x}}} \nonumber \\
        U_{A_x,2}:\ket{i_u}\ket{0}&\mapsto\ket{i_u}\Ket{\frac{-V+(i_u+1)\Delta_\mathbf{v}}{2\Delta_\mathbf{x}}},
    \end{align}
    which are implemented with some arithmetic circuits.
    In the last arrow, the operation at the first arrow is uncomputed.
    The operation in Eq.~(\ref{eq:OAx}) is in fact $O^{A_x}_\mathrm{ent}$, with some registers regarded as ancillary.
    $O^{A_y}_\mathrm{ent}$ and $O^{A_z}_\mathrm{ent}$ are implemented similarly.
    
    To implement $O^{A_u(t)}_\mathrm{ent}$ for any $t\in[0,T)$, it is sufficient to implement $O^{A_u^{i_t}}_\mathrm{ent}$ for any $i_t\in[n_t]_0$, where $A_u^{i_t}:=A_u(t_{i_t})$.
    This is done as follows:
    \begin{align}
        &\ket{i}\ket{j}\ket{0}\ket{0}\ket{0} \nonumber \\
        \rightarrow & \ket{i}\ket{j}\left(\mathbf{1}_{\sigma^{-1}(i)-\sigma^{-1}(j)\ne(\pm1,0,0,0,0,0)}\ket{0}\ket{0}\right.  \nonumber \\
        & \qquad \quad +\mathbf{1}_{\sigma^{-1}(i)-\sigma^{-1}(j)=(1,0,0,0,0,0)}\ket{1}\ket{0} \nonumber \\
        & \qquad \quad \left.+\mathbf{1}_{\sigma^{-1}(i)-\sigma^{-1}(j)=(-1,0,0,0,0,0)}\ket{0}\ket{1}\right)\ket{0} \nonumber \\
        \rightarrow & \ket{i}\ket{j} \nonumber \\
        & \otimes\left(\mathbf{1}_{\sigma^{-1}(i)-\sigma^{-1}(j)\ne(0,0,0,\pm1,0,0)}\ket{0}\ket{0}\ket{0}\right.  \nonumber \\
        & +\mathbf{1}_{\sigma^{-1}(i)-\sigma^{-1}(j)=(0,0,0,1,0,0)}\ket{1}\ket{0}\Ket{-\frac{F_{\mathrm{CDM},x}^{i_t,\mathbf{i}_\mathbf{x}} }{2\Delta_\mathbf{v}}} \nonumber \\
        & \left.+\mathbf{1}_{\sigma^{-1}(i)-\sigma^{-1}(j)=(0,0,0,-1,0,0)}\ket{0}\ket{1}\Ket{\frac{F_{\mathrm{CDM},x}^{i_t,\mathbf{i}_\mathbf{x}} }{2\Delta_\mathbf{v}}}\right) \nonumber \\
        \rightarrow & \ket{i}\ket{j}\ket{0}\ket{0}\ket{(A_u^{i_t})_{ij}}.
        \label{eq:OAu}
    \end{align}
    Here, the first arrow is similar to Eq.~(\ref{eq:OAx}).
    At the second arrow, we use the controlled versions of
    \begin{align}
        U_{A_u^{i_t},1}:\ket{\mathbf{i}_\mathbf{x}}\ket{0}&\mapsto\ket{\mathbf{i}_\mathbf{x}}\Ket{-\frac{F_{\mathrm{CDM},x}^{i_t,\mathbf{i}_\mathbf{x}}}{2\Delta_\mathbf{v}}} \nonumber \\
        U_{A_u^{i_t},2}:\ket{\mathbf{i}_\mathbf{x}}\ket{0}&\mapsto\ket{\mathbf{i}_\mathbf{x}}\Ket{\frac{F_{\mathrm{CDM},x}^{i_t,\mathbf{i}_\mathbf{x}}}{2\Delta_\mathbf{v}}},
    \end{align}
    which are implemented with $O^{i_t}_{F_{\mathrm{CDM},x}}$ and some arithmetic circuits.
    The last arrow is uncomputation.
    $O^{A_v(t)}_\mathrm{ent}$ and $O^{A_w(t)}_\mathrm{ent}$ are implemented similarly.

    Then, we can perform the following operation
    \begin{align}
        & \ket{i}\ket{j}\ket{0}\ket{0}\ket{0}\ket{0}\ket{0}\ket{0}\ket{0} \nonumber \\
        \rightarrow&\ket{i}\ket{j}\ket{(A_x)_{ij}}\ket{(A_y)_{ij}}\ket{(A_z)_{ij}} \nonumber \\
        & \otimes \ket{(A_u(t))_{ij}}\ket{(A_v(t))_{ij}}\ket{(A_w(t))_{ij}}\ket{0} \nonumber \\
        \rightarrow&\ket{i}\ket{j}\ket{(A_x)_{ij}}\ket{(A_y)_{ij}}\ket{(A_z)_{ij}} \nonumber \\
        & \otimes\ket{(A_u(t))_{ij}}\ket{(A_v(t))_{ij}}\ket{(A_w(t))_{ij}}\ket{(A(t))_{ij}} \nonumber \\
        \rightarrow&\ket{i}\ket{j}\ket{0}\ket{0}\ket{0}\ket{0}\ket{0}\ket{0}\ket{(A(t))_{ij}},
        \label{eq:OAent}
    \end{align}
    where the first arrow is by $O^{A_x}_\mathrm{ent},O^{A_y}_\mathrm{ent},O^{A_z}_\mathrm{ent},O^{A_u(t)}_\mathrm{ent},O^{A_v(t)}_\mathrm{ent}$ and $O^{A_w(t)}_\mathrm{ent}$, we use the adder to compute $(A(t))_{ij}=(A_x)_{ij}+(A_y)_{ij}+(A_z)_{ij}+(A_u(t))_{ij}+(A_v(t))_{ij}+(A_w(t))_{ij}$ at the second arrow, and the last arrow is uncomputation.
    This is nothing but $O^{A(t)}_\mathrm{ent}$.

    To complete the proof, we note that the number of the queries to $\left\{O^{i_t}_{F_{\mathrm{CDM},x}}, O^{i_t}_{F_{\mathrm{CDM},y}}, O^{i_t}_{F_{\mathrm{CDM},z}}\right\}_{i_t}$ and arithmetic circuits in the operations (\ref{eq:OAx}) and (\ref{eq:OAu}) are $O(1)$ and so is that in Eq.~(\ref{eq:OAent}), which proves the statement on the query number.
\end{proof}

Then, let us prove Theorem \ref{th:ketfTGen}.

\begin{proof}[Proof of Theorem \ref{th:ketfTGen}]
    First, let us construct the block-encoding of $\exp(-i \Delta_t H_{i_t})$ for any $i_t\in[n_t]_0$.
    Noting that $H_{i_t}$ is $12$-sparse, because of Theorem \ref{th:BlEncSp} and Lemma \ref{lem:OH}, we have $\left(12\|H_{i_t}\|_\mathrm{max},\lg N_\mathrm{gr}+3,\epsilon/2T\right)$-block-encoding $U_{H_{i_t}}$ of $H_{i_t}$, in which $O^{i_t}_{F_{\mathrm{CDM},x}}, O^{i_t}_{F_{\mathrm{CDM},y}}, O^{i_t}_{F_{\mathrm{CDM},z}}$ and arithmetic circuits are queried $O\left(1\right)$ times.
    Then, because of Theorem \ref{th:BlEncExpH}, we have a $(1,\lg N_\mathrm{gr}+5,\epsilon/n_t)$-block-encoding $V_{i_t}$ of $\exp(-i \Delta_t H_{i_t})$, in which (controlled) $U_{H_{i_t}}$ and its inverse are queried
    \begin{equation}
    O\left(\|H_{i_t}\|_\mathrm{max}\Delta_t+\log\left(\frac{n_t}{\epsilon}\right)\right)
    \label{eq:CompVit}
    \end{equation}
    times, and thus so are (controlled) $O^{i_t}_{F_{\mathrm{CDM},x}}, O^{i_t}_{F_{\mathrm{CDM},y}}, O^{i_t}_{F_{\mathrm{CDM},z}}$, arithmetic circuits and their inverses.

    Then, we generate $\ket{0}^{\otimes(\lg N_\mathrm{gr}+5)}\ket{{\mathbf{f}(0)}}$ by using $O_{\mathbf{f}(0)}$ and apply $V_0$ to it.
    This yields
    \begin{equation}
        \ket{\phi_1}:=\ket{0}^{\otimes(\lg N_\mathrm{gr}+5)}\exp(-i \Delta_t H_{0})\ket{{\mathbf{f}(0)}} + \ket{\psi_{\mathrm{gar},1}},
    \end{equation}
    where $\ket{\psi_{\mathrm{gar},1}}$ is an unnormalized quantum state satisfying $\|\ket{\psi_{\mathrm{gar},1}}\|\le \epsilon/n_t$.
    We further apply $V_1$ to $\ket{\phi_1}$ and get
    \begin{equation}
        \ket{\phi_2}:=\ket{0}^{\otimes(\lg N_\mathrm{gr}+5)}e^{-i \Delta_t H_{1}}e^{-i \Delta_t H_{0}}\ket{{\mathbf{f}(0)}}+\ket{\psi_{\mathrm{gar},2}}.
    \end{equation}
    Here,
    \begin{equation}
        \ket{\psi_{\mathrm{gar},2}}:=e^{-i \Delta_t H_{1}}\ket{\psi_{\mathrm{gar},1}}+\ket{\psi_{\mathrm{gar},2}^\prime},
    \end{equation}
    where $\ket{\psi_{\mathrm{gar},2}^\prime}$ is an unnormalized state with norm at most $\epsilon/n_t$.
    The norm of $\ket{\psi_{\mathrm{gar},2}}$ is bounded as
    \begin{equation}
        \|\ket{\psi_{\mathrm{gar},2}}\|\le \|e^{-i \Delta_t H_{1}}\ket{\psi_{\mathrm{gar},1}}\|+\|\ket{\psi_{\mathrm{gar},2}^\prime}\|\le \frac{2\epsilon}{n_t}.
    \end{equation}
    Continuing this, we obtain
    \begin{align}
        &V_{n_t-1} \cdots V_0 \ket{0}^{\otimes(\lg N_\mathrm{gr}+5)}\ket{{\mathbf{f}(0)}} \nonumber \\
        =& \ket{0}^{\otimes(\lg N_\mathrm{gr}+5)}e^{-i \Delta_t H_{n_t-1}} \cdots e^{-i \Delta_t H_{0}}\ket{{\mathbf{f}(0)}}+\ket{\psi_{\mathrm{gar},n_t}},
        \label{eq:Vsf0}
    \end{align}
    where the first term is nothing but
    \begin{equation}
        \ket{0}^{\otimes(\lg N_\mathrm{gr}+5)}\ket{{\mathbf{f}(T)}}
    \end{equation}
    because of Eq.~(\ref{eq:HPieceConst}), the piecewise time-constancy of $H(t)$, and $\ket{\psi_{\mathrm{gar},n_t}}$ is the unnormalized state with norm at most $\frac{\epsilon}{n_t} \times n_t=\epsilon$.
    This means that the above operation is $U_{\mathbf{f}(T),\epsilon}$ with the stated property.

    In this operation, the number of queries to (controlled) $\left\{O^{i_t}_{F_{\mathrm{CDM},x}}, O^{i_t}_{F_{\mathrm{CDM},y}}, O^{i_t}_{F_{\mathrm{CDM},z}}\right\}_{i_t}$, arithmetic circuits and their inverses are that in each $V_{i_t}$, which is bounded as Eq.~(\ref{eq:CompVit}), multiplied by $n_\mathrm{t}$, that is,
    \begin{equation}
    O\left(n_\mathrm{t}\left(\|H_{i_t}\|_\mathrm{max}\Delta_t+\log\left(\frac{n_t}{\epsilon}\right)\right)\right)
    \end{equation}
    Replacing $\|H_{i_t}\|_\mathrm{max}$ with its bound
    \begin{align}
        &\|H_{i_t}\|_\mathrm{max} \nonumber \\
        \le & \max\left\{\max_{i_u\in[n_\mathrm{gr}]_0}\left|\frac{u_{i_u}}{2\Delta_\mathbf{x}}\right|,\max_{i_v\in[n_\mathrm{gr}]_0}\left|\frac{v_{i_v}}{2\Delta_\mathbf{x}}\right|,\right.  \max_{i_w\in[n_\mathrm{gr}]_0}\left|\frac{w_{i_w}}{2\Delta_\mathbf{x}}\right|,\nonumber \\
        &\qquad \quad \left.\max_{\mathbf{i}_\mathbf{x}\in\mathcal{I}_3}\left|\frac{F_{\mathrm{CDM},x}^{i_t,\mathbf{i}_\mathbf{x}}}{2\Delta_\mathbf{v}}\right|,\max_{\mathbf{i}_\mathbf{x}\in\mathcal{I}_3}\left|\frac{F_{\mathrm{CDM},y}^{i_t,\mathbf{i}_\mathbf{x}}}{2\Delta_\mathbf{v}}\right|,\max_{\mathbf{i}_\mathbf{x}\in\mathcal{I}_3}\left|\frac{F_{\mathrm{CDM},z}^{i_t,\mathbf{i}_\mathbf{x}}}{2\Delta_\mathbf{v}}\right|\right\} \nonumber \\
        \le & \max\left\{\frac{V}{L},\frac{F_\mathrm{max}}{2V}\right\} \times \frac{n_\mathrm{gr}}{2}
        \label{eq:Hmax}
    \end{align}
    yields Eq.~(\ref{eq:CompfTGen}).

    Lastly, we prove the statement on the qubit number.
    Combining Theorems \ref{th:BlEncSp} and \ref{th:BlEncExpH}, we see that the number of the ancilla qubits used in $V_{i_t}$ is
    \begin{equation}
        O\left(\log^{5/2}\left(\frac{\|H_{i_t}\|_\mathrm{max}}{\epsilon/2T} \right)\right).
    \end{equation}
    Plugging Eq. (\ref{eq:Hmax}) into this, we get Eq. (\ref{eq:QubitfTGen}).
    Adding the qubits on which the quantum state in Eq. (\ref{eq:Vsf0}) is generated, whose number is $2\mathrm{lg} N_\mathrm{gr}+5=O(\log n_\mathrm{gr})$, does not change the order. 
\end{proof}

\subsection{Proof of Theorem \ref{th:PnuEstim} \label{sec:PrThPnuEstim}}

\begin{proof}[Continuation of the proof]

    Let us see that the stated accuracy is achieved.
    By $U_{\mathbf{f}(T),C\epsilon/4}$, we get the following state
    \begin{equation}
        U_{\mathbf{f}(T),C\epsilon/4}\ket{0}^{\otimes(2\lg N_\mathrm{gr}+5)}=\ket{0}^{\otimes(\lg N_\mathrm{gr}+5)}\ket{{\mathbf{f}(T)}} + \ket{\phi_\mathrm{gar}},
    \end{equation}
    where $\ket{\phi_\mathrm{gar}}$ is an unnormalized state with $\|\ket{\phi_\mathrm{gar}}\|\le C\epsilon/4$.
    Thus, defining
    \begin{equation}
        p:=\left|\Bra{\Omega^\prime_{\mathbf{i}_\mathbf{k}}}W^\prime U_{\mathbf{f}(T),C\epsilon/4}\ket{0}^{\otimes(2\lg N_\mathrm{gr}+5)}\right|^2,
        \label{eq:EstimProb}
    \end{equation}
    we obtain
    \begin{align}
        & \left|p-C\left|\tilde{\delta}^\nu_{\mathbf{i}_\mathbf{k}}\right|^2\right| \nonumber \\
        =&\left| 2\Re \Bra{\Omega_{\mathbf{i}_\mathbf{k}}}W\ket{\mathbf{f}(T)}\Bra{\Omega^\prime_{\mathbf{i}_\mathbf{k}}}W^\prime\ket{\phi_\mathrm{gar}}+ \left|\Bra{\Omega^\prime_{\mathbf{i}_\mathbf{k}}}W^\prime\ket{\phi_\mathrm{gar}}\right|^2\right| \nonumber \\
        \le & 2\frac{C\epsilon}{4} + \left(\frac{C\epsilon}{4}\right)^2 \nonumber \\
        \le & \frac{3C\epsilon}{4}.
        \label{eq:pErr}
    \end{align}
    Here, at the first inequality, we use $|\Bra{\Omega_{\mathbf{i}_\mathbf{k}}}W\ket{\mathbf{f}(T)}|\le1$ and $\|\Bra{\Omega^\prime_{\mathbf{i}_\mathbf{k}}}W^\prime\ket{\phi_\mathrm{gar}}\|\le\|(W^\prime)^\dagger\Ket{\Omega^\prime_{\mathbf{i}_\mathbf{k}}}\|\cdot\|\ket{\phi_\mathrm{gar}}\|\le\frac{C\epsilon}{4}$, which follows from the Cauchy--Schwarz inequality.
    $\tilde{p}$ obtained in Step 3 in Algorithm \ref{alg:PnuEstim} is an $C\epsilon/4$-approximation of $p$, and combining this with Eq.~(\ref{eq:pErr}) yields
    \begin{align}
        &\left|\frac{\tilde{p}}{C}-\left|\tilde{\delta}^\nu_{\mathbf{i}_\mathbf{k}}\right|^2\right| \nonumber \\
        \le & \left|\frac{\tilde{p}}{C}-\frac{p}{C}\right|+\left|\frac{p}{C}-\left|\tilde{\delta}^\nu_{\mathbf{i}_\mathbf{k}}\right|^2\right| \nonumber \\
        \le & \frac{\epsilon}{4} + \frac{3\epsilon}{4} \nonumber \\
        \le & \epsilon.
    \end{align}
    Thus, the error of the output of Algorithm \ref{alg:PnuEstim} is bounded by $\epsilon$.

    Lastly, let us prove the statement on the query complexity.
    In the QAE in Step 3 in Algorithm \ref{alg:PnuEstim}, $U_{\mathbf{f}(T),C\epsilon/4}$ is queried
    \begin{equation}
        O\left(\frac{1}{C\epsilon}\log\left(\frac{1}{\delta}\right)\right)
        \label{eq:CompUfTinQAE}
    \end{equation}
    times.
    Because of Theorem \ref{th:ketfTGen}, in one call to $U_{\mathbf{f}(T),C\epsilon/4}$, the number of queries to (controlled) $\left\{O^{i_t}_{F_{\mathrm{CDM},x}}, O^{i_t}_{F_{\mathrm{CDM},y}}, O^{i_t}_{F_{\mathrm{CDM},z}}\right\}_{i_t}$, arithmetic circuits and their inverses is
    \begin{equation}
        O\left(n_\mathrm{gr}T\times\max\left\{\frac{V}{L},\frac{F_\mathrm{max}}{V}\right\} + n_t\log\left(\frac{n_t}{C\epsilon}\right) \right),
        \label{eq:CompOrainUfT}
    \end{equation}
    and that to $O_{\mathbf{f}(0)}$ is 1.
    Then, combining these, we bound the total query number in Algorithm \ref{alg:PnuEstim} as Eqs.~(\ref{eq:CompfPnuEstim1}) and (\ref{eq:CompfPnuEstim2}).

    On the qubit number, we note that, to operate $U_{\mathbf{f}(T),C\epsilon/4}$, we use qubits whose number if of order (\ref{eq:QubitPnuEstim1}), as implied by Theorem \ref{th:ketfTGen}.
    Since operating $W$ and QAE do not require additional qubits, the number of qubits used in the whole of Algorithm \ref{alg:PnuEstim} is also of order (\ref{eq:QubitPnuEstim1}).

\end{proof}

\bibliography{reference}

\end{document}